\documentclass[a4paper]{article}
\usepackage[left=2cm,right=2cm, top=2.5cm,bottom=2.5cm,bindingoffset=0cm]{geometry}
\usepackage{tikz, titlesec}
\usetikzlibrary{calc,intersections,through,backgrounds}
\usepackage{amsmath, amssymb, amsthm}

\usepackage[all]{xy}

\usepackage{array, hyperref}

\titleformat{\section}{\large\bfseries\filcenter}{\thesection}{1em}{}
\titleformat{\subsection}{\bfseries}{\thesubsection}{1em}{}
\titleformat{\subsubsection}[runin]{\bfseries}{\thesubsubsection}{1em}{}[.]

\usetikzlibrary{arrows} 
\tikzset{
    >=stealth',
    pil/.style={
           ->,
           thick,
           shorten <=2pt,
           shorten >=2pt,}
}

\usepackage{url}
\usepackage{color}
\usepackage{mathtools}                              
\usepackage{caption}

\usepackage{tikz}
\usetikzlibrary{calc,intersections,through,backgrounds}
\usetikzlibrary{arrows,decorations.markings}

\newtheorem{thm}{Theorem}

\newtheorem{prop}[thm]{Proposition}
\newtheorem{cor}[thm]{Corollary}

\theoremstyle{defi}

\theoremstyle{rem}
\newtheorem{rem}[thm]{Remark}
\newtheorem{exa}[thm]{Example}
\numberwithin{equation}{section}
\numberwithin{thm}{section}

\newcommand{\n}{\mathfrak{n}}

\renewcommand{\b}{\mathfrak{b}}

\newcommand{\gl}{\mathfrak{gl}}

\newcommand{\sli}{\mathfrak{sl}}

\newcommand{\C}{\mathbb{C}}

\newcommand{\Z}{\mathbb{Z}}
\newcommand{\N}{\mathbb{N}}

\DeclareMathOperator{\ib}{\left(\cdot\,,\cdot\right)}
\DeclareMathOperator{\pb}{\lbrace\cdot\,,\cdot\rbrace}

\DeclareMathOperator{\GL}{GL}
\DeclareMathOperator{\SL}{SL}
\DeclareMathOperator{\tr}{Trace}
\DeclareMathOperator{\sgn}{sgn}
\DeclareMathOperator{\diag}{diag}

\newcommand{\nl}[2][]{\left\{#1#2\right\}_{nl}}

\DeclareMathOperator*{\equals}{=}

\providecommand{\keywords}[1]{\textbf{Key words:} #1}
\providecommand{\subject}[1]{\textbf{2010 Mathematics Subject Classification:} #1}

\newpage

\title{A non-local Poisson bracket for Coxeter--Toda lattices}
\author{Vizarreta E.D.C\thanks{
Departmento of Matematica, Universidade Federal de Pernambuco, e-mail: \tt{edanchv@gmail.com}}}
\date{}

\begin{document}
\newpage 

\maketitle
\abstract{We present a non-local Poisson bracket defined on the phase space $G^{u,v}/H$ of a Coxeter--Toda lattice, where $G^{u,v}$ is a Coxeter double Bruhat cell of $\GL_n$ and $H$ is the subgroup of diagonal matrices. This non-local Poisson bracket is given in an appropriate set of coordinates of $G^{u,v}/H$ derived from the so-called {\it factorization parameters}. We prove that the generalized B\"acklund--Darboux transformations $\sigma_{u,v}^{u',v'}: G^{u,v}/H\to G^{u',v'}/H$ are Poisson maps. We exploit that fact to show that the non-local Poisson bracket corresponds to the Atiyah--Hitchin bracket under the Moser map. 
}
\vspace{0.2cm}

\keywords{Coxeter--Toda lattice; compatible Poisson brackets; Atiyah--Hitchin bracket; Moser map; generalized B\"acklaund--Darboux transformations}

\subject{37K10; 53D17}



\section{Introduction} \label{intro}

One of the main features shared by most finite dimensional completely integrable Hamiltonian systems is the existence of a bi-Hamiltonian structure, i.e. the dynamical system admits two Hamiltonian formulations with respect to a pair of compatible Poisson brackets defined on the underlying phase space.


Two of the most renowned and well documented examples of finite dimensional completely integrable Hamiltonian systems are the Toda lattice and the relativistic Toda lattice. The former was discovered in \cite{To} as a special case of the famous Fermi--Pasta--Ulam lattice. The latter was introduced in \cite{Ruis}. The bi-Hamiltonian nature of the Toda and the relativistic Toda lattice is one of the most classical results. For instance, Suris \cite{Su} has presented both lattices as restrictions of the same Hamiltonian system to different symplectic leaves.

Moser, in his celebrated paper \cite{Mo}, made a fundamental contribution to the study of solutions of the non-periodic Toda lattice. To linearize the Toda lattice, he defined a map from the space of finite Jacobi matrices to the space of proper rational functions of fixed degree. The map, nowadays called the {\it Moser map}, associates to each Jacobi matrix its respective {\it Weyl function}, which is a matrix entry of the resolvent of the Jacobi matrix. The inverse problem, i.e. the reconstruction of the Jacobi matrix from its Weyl function, was a fundamental key on Moser's work to solve the non-periodic Toda lattice. 

The Toda and the relativistic Toda lattice were described in \cite{FG1} as particular lattices of a class of integrable lattices derived from the full Kostant--Toda flows on Hessenberg matrices, called {\it elementary Toda lattices}. Each of these systems can be linearized using the Moser map.  The inverse problem for elementary Toda lattices was solved in \cite{FG3}.

In the study of magnetic monopoles a natural Poisson structure, called the {\it Atiyah--Hitchin bracket}, on the space of rational functions of fixed degree has been introduced in \cite{A-H}. Later, Faybusovich and Gekhtman \cite{FG2} noticed that the Atiyah--Hitchin bracket fits into a finite family of compatible Poisson brackets on the space of rational functions. The latter family along with the Moser map established the multi-Hamiltonian structure of the elementary Toda lattices.

Elementary Toda lattices belong to a broader family of integrable lattices, the so-called {\it Coxeter--Toda lattices}.  Toda flows on $\GL_n$ are commuting Hamiltonian flows generated by conjugation-invariant functions on $\GL_n$ with respect to the standard Poisson--Lie structure. Toda flows for an arbitrary standard semisimple Poisson--Lie group were studied in \cite{Resh}. A Coxeter--Toda lattice is an induced flow on $G^{u,v}/H$ from the restriction of a Toda flow to $G^{u,v}$, where $G^{u,v}$ is an instance of a particular class of double Bruhat cells of $\GL_n$, named as {\it Coxeter double Bruhat cell}, and $H$ is the subgroup of diagonal matrices. The term Coxeter--Toda lattice for an arbitrary simple Lie group was coined in \cite{HKKR}. Double Bruhat cells for any semisimple Lie group along with the concepts of {\it factorization parameters} and {\it twisted generalized minors} showed up in \cite{Fo-Ze} in the context of total positivity. Their connections with  integrable systems were explained in \cite{KoZ}.

In \cite{GSV4}, Coxeter--Toda lattices and B\"acklund--Darboux transformations $\sigma_{u,v}^{u',v'}:G^{u,v}/H\to G^{u',v'}/H$ were described from the cluster algebra and the annular weighted networks perspective using a special set of coordinates derived from the factorization parameters of a Coxeter double Bruhat cell $G^{u,v}$. These coordinates can be restored from the Weyl function of generic elements of a Coxeter double Bruhat cell. Thus the Moser map for Coxeter--Toda lattices is invertible.  Once again, the Faybusovich--Gekhtman family can be used to establish the multi-Hamiltonian nature of the Coxeter--Toda lattices. 

Although the multi-Hamiltonian nature of Coxeter--Toda lattices is well known, apart from the quadratic Poisson bracket \eqref{eq: quadratic epsilon} computed in \cite{GSV4}, no other explicit Poisson bracket and compatible with \eqref{eq: quadratic epsilon}, in the set of coordinates of $G^{u,v}/H$ mentioned above, did appear in the literature. The non-local Poisson bracket $\pb_{nl}$ on $G^{u,v}/H$ presented in \eqref{eq: rat 1}--\eqref{eq: rat 4} addresses this issue (see Corollary \eqref{cor: 2}). To this end, besides showing that the non-local bracket is Poisson, we will need that the generalized B\"acklund--Darboux transformations preserve the non-local Poisson bracket. 

We deduced the non-local Poisson bracket $\pb_{nl}$ on $G^{u,v}/H$ after performed Maple computations of the Moser map $m_{u,v}:G^{u,v}/H\to (\mathcal{W}_n,\pb_0)$ for small values of $n$, where $\mathcal{W}_n$ is a subset of the space of rational functions of fixed degree and $\pb_0$ is the Atiyah--Hitchin bracket. Therefore, as is expected, we show that the non-local Poisson bracket corresponds to the Atiyah--Hitchin under the Moser map, Corollary \ref{cor: 1}

This manuscript is organized as follows. Section 2 is devoted to a review of the concepts of factorization parameters for double Bruhat cells $G^{u,v}$, the standard Poisson--Lie structure in the $\GL_n$ case and the family of Faybusovich--Gekhtman Poisson brackets. Section 3 treats the Coxeter--Toda lattices and the generalized B\"acklund--Darboux transformations among them. Our main results are in Section 4. First, we present a non-local bracket $\pb_{nl}$ on $G^{u,v}/H$, in coordinates derived from the factorization parameters, and show that it is a Poisson bracket, Theorem \ref{thm: 1}. When $u^{-1}=v=s_{n-1}\cdots s_1$, after a change of variables, the non-local Poisson bracket becomes the well known linear Poisson bracket \eqref{eq: linear Toda bracket} used for a Hamiltonian formulation of the Toda lattice. Therefore, by Proposition \ref{prop: 1}, when $u^{-1}=v=s_{n-1}\cdots s_1$, the non-local bracket $\pb_{nl}$ corresponds to the Atiyah--Hitchin bracket under the Moser map. We then present Theorem \ref{prop: non local poisson maps}, the proof is sketched in Appendix \ref{Appendix A}, to draw our conclusions of interest mentioned above. The proofs are matter of direct, tough lengthy, computation. We finish with two questions that we would like to address in the future




\section{Preliminaries}

Let us first recall the factorization parameters of a double Bruhat cell, the standard Poisson--Lie bracket and the Faybusovich--Gekhtman family which includes the Atiyah--Hitchin bracket.

\subsection{Double Bruhat cells and factorization parameters}

Let $\b_+,\n_+,\b_-$ and $\n_-$ be the algebras of upper triangular, strictly upper triangular, lower triangular and strictly lower triangular matrices, respectively. The connected subgroups that correspond to $\b_+,\n_+,\b_-$ and $\n_-$ will be denoted by $B_+,B_-,N_+$ and $N_-$, respectively. We denote by $H$ the subgroup of diagonal matrices in $\GL_n$.

Since every $\xi\in\gl_n$ has a unique decomposition into
\[
\xi=\xi_-+\xi_0+\xi_+,\quad \xi_{\pm}\in\n_{\pm},\,\xi_0\text{ diagonal},
\]
then every $X$ in an open Zariski dense subset of $\GL_n$ has a unique Gauss decomposition:
\[
X=X_-X_0X_+,\quad X_{\pm}\in N_{\pm},\,X_0\in H.
\]
Let $s_i,\,i\in[1,n-1]$, denote the elementary transposition $(i,i+1)$ in the symmetric group $S_n$. A {\it reduced decomposition}, not in a unique way, of an element $w\in S_n$ is a representation of $w$ as a product $w=s_{i_1}\cdots s_{i_l}$ of the smallest possible length. The number $l$ depends only on $w$, it will be denoted by $l(w)$, and is called the {\it length} of $w$. The sequence of indices $(i_1,\ldots,i_l)$ that corresponds to a given reduced decomposition of $w$ is called a {\it reduced word} for $w$. The notion of a reduced word for an ordered pair $(u,v)$ of elements in $S_n$ is defined as follows: if $(i_1,\ldots,i_{l(u)})$ is a reduced word for $u$ and $(i'_1,\ldots,i'_{l(v)})$ is a reduced word for $v$, then any shuffle of sequences $(i_1,\ldots,i_{l(u)})$ and $(-i'_1,\ldots,-i'_{l(v)})$ is called a reduced word for $(u,v)$. 

Hereafter, $v\in S_n$ means $v=(\delta_{iv(j)})_{i,j=1}^n$ whenever there is no confusion. The {\it Bruhat decompositions} of $\GL_n$ with respect to $B_+$ and $B_-$ are defined, respectively, by 
\[
\GL_n=\bigcup_{u\in S_n}B_+uB_+,\qquad \GL_n=\bigcup_{v\in S_n}B_-vB_-.
\]

The sets $B_+uB_+$ (respectively, $B_-vB_-$) are called {\it Bruhat cells} (respectively, {\it opposite Bruhat cells}). For any $u,\,v\in S_n$, the {\it double Bruhat cell} $G^{u,v}\subset\GL_n$, first introduced in \cite{Fo-Ze}, is defined as the intersection of a Bruhat cell and an opposite Bruhat cell, i.e.
\[
G^{u,v}=B_+uB_+\cap B_-vB_-.
\]
Note that any double Bruhat cell $G^{u,v}$ is invariant under left and right multiplication by elements of $H$.

\begin{exa}\label{exa: Coxeter DBC} 
{\em
Let $u^{-1}=v=s_{n-1}\cdots s_1$. Then
\[
G^{u,v}=\lbrace (x_{ij})\in\GL_n\mid x_{ij}=0\,\text{ if }\vert i-j\vert>1\text{ and }\prod_{i=1}^{n-1}x_{i,i+1}x_{i+1,i}\neq 0\rbrace.
\]
}
\end{exa}

It follows from \cite[Theorem 1.1]{Fo-Ze} that the variety $G^{u,v}$ is biregularly isomorphic to a Zariski open subset of $\C^{l(u)+l(v)+n}$. Different birational maps $X_{{\bf j}}$ from $\C^{l(u)+l(v)+n}$ to $G^{u,v}$ can be constructed in a explicit way. Namely, let ${\bf j}=(j_1,\ldots,j_{l(u)+l(v)+n})$ be a shuffle of a reduced word ${\bf i}$ for $(u,v)$ and any re-arrangement of the sequence ${\bf k}=\lbrace \sqrt{-1},2\sqrt{-1}\ldots,n\sqrt{-1}\rbrace$; and set 
\[
\theta(j_l)=
\begin{cases}
+ &\text{ if }j_l>0,\\
- &\text{ if }j_l<0,\\
0 &\text{ if }j_l\in{\bf k}.
\end{cases}
\]

For $t\in\C,\,i,j\in[1,n-1]$ and $k\in[1,n]$, let
\[
E_i^-(t)=I_n+te_{i+1,i},\quad E_j^+(t)=I_n+te_{j,j+1},\quad\text{and}\quad E_k^0(t)=I_n+(t-1)e_{k,k}
\]
where $e_{i,j}$ denotes the elementary matrix $(\delta_{i\alpha}\delta_{j\beta})_{\alpha,\beta=1}^n$. 
Then the map $X_{{\bf j}}:\C^{l(u)+l(v)+n}\to G^{u,v}$ can be defined by
\begin{equation}\label{eq: birational map} 
X_{{\bf j}}({\bf t})=\prod_{q=1}^{l(u)+l(v)+n}E_{\vert j_q\vert}^{\theta(j_q)}(t_q).
\end{equation}
The parameters $t_1,\ldots,t_{l(u)+l(v)+n}$ constituting ${\bf t}$ are called {\it factorization parameters}. 
\begin{rem}
{\em
The bidiagonal factorization \eqref{eq: birational map} has an interpretation in terms of directed graphs embedded in a disc with weighted edges. These are called {\it perfect networks} in a disc (see \cite{GSV1}).
}
\end{rem}

\subsection{The standard Poisson--Lie bracket}

Recall now that a Poisson--Lie group is a Lie group $G$ equipped with a Poisson bracket such that the multiplication map $G\times G\to G$ is Poisson.

The {\it standard Poisson--Lie bracket} on $\SL_n$, denoted by $\pb_{\SL_n}$, is defined as follows
\begin{equation*}
\lbrace f_1,f_2\rbrace_{\SL_n}(X)=\frac{1}{2}(R(\nabla f_1(X)X),\nabla f_2(X)X)-\frac{1}{2}(R(X\nabla f_1(X)),X\nabla f_2(X))
\end{equation*}
where $\ib$ denotes the trace form of $\sli_n$, $\nabla$ is the gradient defined with respect to the trace form, and $R:\sli_n\to\sli_n$ is the standard $R$-matrix
\[
R(\xi)=(\sgn(j-i)\xi_{ij})_{i,j=1}^n
\]
On coordinate functions $x_{ij},\,x_{kl}$, the bracket $\pb_{\SL_n}$ verify the following relation
\begin{equation}\label{eq: standard P-L SL_n} 
\lbrace x_{ij},x_{kl}\rbrace_{\SL_n}=\frac{1}{2}(\sgn(k-i)+\sgn(l-j))x_{il}x_{kj}.
\end{equation}

The standard Poisson--Lie structure on $\GL_n$, denoted by $\pb_{\GL_n}$, is the Poisson bracket \eqref{eq: standard P-L SL_n} on $\GL_n$ considering the determinant function on $\GL_n$ as a Casimir function.

\begin{rem}
{\em
The notion of standard Poisson--Lie bracket is defined for any semisimple Lie group $G$. Kogan and Zelevinsky \cite[Theorem 2.3]{KoZ} described the symplectic leaves of such Poisson structure as translations by elements of the Cartan subgroup of a particular symplectic leaf inside of a double Bruhat cell.
}
\end{rem}

In particular, if $\GL_n$ is equipped with $\pb_{\GL_n}$ then every double Bruhat cell $G^{u,v}\subset\GL_n$ is a regular Poisson submanifold. Symplectic leaves in $\left(\GL_n,\pb_{\GL_n}\right)$ are of the form $S^{u,v}\cdot a$, where $S^{u,v}\subset G^{u,v}$ is a distinguished symplectic leaf and $a$ is an element of $H$. Furthermore, the dimension of the symplectic leaves in $G^{u,v}$ are equal to $l(u)+l(v)+\text{corank}(uv^{-1}-I_n)$. 

\begin{rem} 
{\em
In \cite{GSV1}, the standard Poisson--Lie bracket on $\GL_n$ and the network representation in a disc of \eqref{eq: birational map} led  to the notion of {\it standard Poisson bracket} on the space of edge weights and on the space of face weights of a network in a disc. Later, this notion of standard Poisson bracket was extended for networks embedded in an annulus (see \cite{GSV2}).
} 
\end{rem}

\subsection{The Faybusovich--Gekhtman Poisson brackets}

Consider a space of rational functions
\[
\text{Rat}_n=\left\{s(\lambda)=\frac{q(\lambda)}{p(\lambda)}\mid p(\lambda)\text{ is a monic polynomial},\,\deg p=n,\,\deg q<n \right\}.
\]
For fixed $p(\lambda),q(\lambda)$ and $k=0,\ldots,n-1$, let us denote 
\begin{equation*}
q^{[k]}(\lambda)=\lambda^{k}q(\lambda)\qquad(\text{mod}\,p(\lambda))
\end{equation*}
and define a skew-symmetric bracket $\pb_{k}$ on the coefficients of $p(\lambda),q(\lambda)$ by setting
\begin{equation}\label{eq: multibrackets} 
\begin{aligned}
\left\{p(\lambda),p(\mu)\right\}_k&=\left\{q(\lambda),q(\mu)\right\}_{k}=0,\\
\left\{p(\lambda),q(\mu)\right\}_k&=\frac{p(\lambda)q^{[k]}(\mu)-p(\mu)q^{[k]}(\lambda)}{\lambda-\mu}.
\end{aligned}
\end{equation}

The bracket $\pb_0$ of the family \eqref{eq: multibrackets} is known as the Atiyah--Hitchin bracket \cite{A-H}.

\begin{prop}[{\cite[Proposition 2]{FG2}}]\label{prop: compatible} 
The brackets $\pb_k\;(k=0,\ldots,n-1)$ are compatible Poisson brackets on $\text{Rat}_n$.
\end{prop}

Consider now
\[
\text{Rat}'_n=\left\{\frac{q(\lambda)}{p(\lambda)}\in\text{Rat}_n\left.\right| q(\lambda) \text{ is monic }\right\}
\]
a subset of $\text{Rat}_n$. The Poisson brackets $\pb_k$ $(k=0,\ldots,n-1)$ given by \eqref{eq: multibrackets} can be restricted to $\text{Rat}'_{n}$ and take the form 
\begin{equation}\label{eq: reduced F-G} 
\begin{aligned}
\lbrace p(\lambda),p(\mu)\rbrace_k&=\lbrace q(\lambda),q(\mu)\rbrace_k=0,\\
\lbrace p(\lambda),q(\mu)\rbrace_k&=\frac{p(\lambda)q^{[k]}(\mu)-p(\mu)q^{[k]}(\lambda)}{\lambda-\mu}-q^{[k]}(\lambda)q(\mu).
\end{aligned}
\end{equation}

\begin{exa}
{\em 
Let $n=2$. In this case we have $p(\lambda)=\lambda^2+\alpha_1\lambda+\alpha_0$, $q(\lambda)=\lambda+\beta_0$
\begin{align*}
\lbrace p(x),q(y)\rbrace_0&=\frac{p(x)q(y)-p(y)q(x)}{x-y}-q(x)q(y)=\alpha_1\beta_0-\alpha_0-\beta_0^2.
\end{align*}
Therefore, we obtain
\begin{align}\label{eq: AH n=2} 
\lbrace \alpha_0,\beta_0\rbrace_0=\alpha_1\beta_0-\alpha_0-\beta_0^2, \quad \lbrace \alpha_1,\beta_0\rbrace_0=\lbrace \alpha_1,\alpha_0\rbrace_0=0.
\end{align}
}
\end{exa}

\section{Coxeter--Toda lattices on $\GL_n$}

The main reference for this section is \cite{GSV4}.

\subsection{Toda flows}

{\it Toda flows} on $\GL_n$ are equations of motion generated by $F_k(X)=\frac{1}{k}\tr(X^k)$ $(k=1,\ldots,n-1)$ and the standard Poisson--Lie structure $\pb_{\GL_n}$. The equation of motion generated by $F_k$ has the following Lax form
\begin{equation}\label{eq: Toda flow} 
\dot X=\left[X,-\frac{1}{2}(\pi_{+}(X^k)-\pi_{-}(X^k))\right]
\end{equation}
where $\pi_+(A)$ and $\pi_-(A)$ denote strictly upper and lower parts of a matrix $A$. The functions $F_1,\ldots,F_{n-1}$ form a maximal family of algebraically independent, conjugation-invariant, functions of $\GL_n$ and they Poisson commute.

Since the action of $H$ on $\GL_n$ by conjugation is Poisson with respect to the standard Poisson--Lie structure and preserves double Bruhat cells, the standard Poisson--Lie structure induces a Poisson structure on $G^{u,v}/H$. Therefore, the Toda hierarchy induces a family of commuting Hamiltonian flows on $G^{u,v}/H$. 

\begin{exa}
{\em
Let $u$ and $v$ as in Example \ref{exa: Coxeter DBC}. Then the space $G^{u,v}/H$ is described as the set of Jac {\it Jacobi matrices} of the form
\begin{equation*}
L=\begin{pmatrix}
b_1 & 1 & 0 & \ldots & 0\\
a_1 & b_2 & 1 & \ldots &0\\
\vdots & \ddots & \ddots & \ddots & \vdots\\
0 &\ldots & a_{n-2} & b_{n-1} & 1\\
0 & \ldots & 0 & a_{n-1} & b_n
\end{pmatrix},\quad a_1\ldots a_{n-1}\neq 0,\quad\det L\neq 0.
\end{equation*}
The Lax equations \eqref{eq: Toda flow} then become the equations of the {\it finite non-periodic Toda hierarchy}
\[\dot L=[L,\pi_-(L^k)].\]
The case $k=1$ is the well-known Toda lattice
\begin{equation}\label{eq: Toda lattice} 
\begin{aligned}
\dot a_j &=a_j(b_{j+1}-b_j),\quad j=1,\ldots,n-1,\\
\dot b_j &=a_j-a_{j-1},\quad j=1,\ldots,n,
\end{aligned}
\end{equation}
with the boundary conditions $a_n=a_0=0$.
}
\end{exa}

Let us consider
\[
\mathcal{W}_n=\left\{\frac{q(\lambda)}{p(\lambda)}\in
\text{Rat}'_n\left.\right|\deg p=\deg q+1,\,p \text{ and }q\text{ are coprime, }p(0)\neq 0\right\}.
\]
a subset of $\text{Rat}'_n$. The {\it Weyl function} of $X\in\GL_n$ is a rational function
\begin{equation}\label{eq: Weyl function} 
m(\lambda)=m(\lambda;X)=\left((\lambda I-X)^{-1}e_1,e_1\right)=\frac{\Delta_{\left[2,n\right]}(\lambda)}{\Delta_{\left[1,n\right]}(\lambda)},
\end{equation}
where $\Delta_{\left[1,n\right]}(\lambda)$ is the characteristic polynomial of $X$, $\Delta_{\left[2,n\right]}(\lambda)$ is the characteristic polynomial of the $(n-1)\times(n-1)$ submatrix of $X$ formed by deleting the first row and column. Since the Weyl function is invariant under the action of $H$ on $G^{u,v}$ by conjugation, one can consider the map
\[
m_{u,v}\colon G^{u,v}/H\to\mathcal{W}_n,
\] 
called {\it Moser map}. In the tridiagonal case, $v=u^{-1}=s_{n-1}\cdots s_1$, Moser \cite{Mo} proved that the map $m_{u,v}$ is invertible and the system \eqref{eq: Toda lattice} is a completely integrable system in the sense of Arnold--Liouville. The level sets of the function $\det L$ foliate Jac into $2(n-1)$-dimensional symplectic manifolds. 

Gekthman, Shapiro and Vainshtein \cite{GSV4} proved that there are other double Bruhat cells, called {\it Coxeter double Bruhat cells}, which share common features with the tridiagonal case. 

\subsection{Coxeter double Bruhat cells on $\GL_n$}

A Coxeter element $w\in S_n$ is a product of $n-1$ distinct transpositions. Given a pair of Coxeter elements $u,\,v\in S_n$, the double Bruhat cell $G^{u,v}$ is called a {\it Coxeter double Bruhat cell} and in this case $\dim G^{u,v}=3n-2$. For any pair $(u,v)$ of Coxeter elements the integrable equations induced on $G^{u,v}/H$ by Toda flows will be called {\it Coxeter--Toda lattices}.

Note that every Coxeter element $v\in S_n$ can be written in the form
\begin{equation*}\label{eq: Coxeter} 
v=s_{[i_{k-1},i_k]}\cdots s_{[i_1,i_2]}s_{[1,i_1]}
\end{equation*}
for some subset $I=\lbrace 1=i_0<i_1<\cdots<i_k=n\rbrace\subset [1,n]$ where $s_{[p,q]}=s_ps_{p+1}\cdots s_{q-1}$ for $1\leq p<q\leq n$. For a pair $(u,v)$ of Coxeter elements let
\begin{equation}\label{eq: index sets Coxeter} 
I^{+}=\lbrace 1=i_0^{+}<\cdots<i_{k^{+}}^{+}=n\rbrace,\quad
I^{-}=\lbrace 1=i_0^-<\cdots<i_{k^-}^-=n\rbrace
\end{equation}
be the subsets of $[1,n]$ that correspond to $v$ and $u^{-1}$, respectively, previously described. 

For a set of complex parameters $c_1^{\pm},\ldots,c_{n-1}^{\pm};d_1,\ldots,d_n$, define the matrices $D=\diag(d_1,\ldots,d_n)$ and
\begin{equation}
C^{+}_j=\sum_{\alpha=i_{j-1}^{+}}^{i_j^{+}-1}c_{\alpha}^{+}e_{\alpha,\alpha+1},\, j\in[1,k^{+}],\quad
C^{-}_j=\sum_{\alpha=i_{j-1}^{-}}^{i_j^{-}-1}c_{\alpha}^{-}e_{\alpha,\alpha+1},\, j\in[1,k^{-}].
\end{equation}

Using the map \eqref{eq: birational map} one has the following result.

\begin{prop}[{\cite[Lemma 3.3]{GSV4}}]\label{prop: factorization} 
A generic element $X\in G^{u,v}$ can be written as
\begin{equation}\label{eq: Coxeter factorization} 
X=(1-C_1^-)^{-1}\cdots (1-C_{k^-}^-)^{-1}D(1-C_{k^+}^+)^{-1}\cdots (1-C_1^+)^{-1}.
\end{equation}
\end{prop}

Fixed a pair $(u,v)$ of Coxeter elements and hence fixed sets $I^+,I^-$ given by \eqref{eq: index sets Coxeter} set 
\begin{equation*}
\begin{aligned}
\varepsilon_i^{\pm}&=
\begin{cases}
0 & \text{ if } i=i_j \text{ for some } 0<j\leq k_{\pm}\\
1 & \text{ otherwise }
\end{cases}
\end{aligned}
\end{equation*}
and
\[
\varepsilon_i=\varepsilon_i^++\varepsilon_i^-,\qquad\varkappa_i=i+1-\sum_{\beta=1}^i\varepsilon_{\beta}.
\]

Using the Laurent expansion of the Weyl function \eqref{eq: Weyl function} at infinity, one has
\[
m(\lambda)=\sum_{j=0}^{\infty}\frac{h_j(X)}{\lambda^{j+1}}
\]
where $h_j(X)=(X^j)_{11}=(X^je_1,e_1),\, j\in\Z$. For any $l\in\Z,\,i\in\N$ consider
\[
\mathcal{H}_i^{(l)}=(h_{\alpha+\beta+l-i-1})_{\alpha,\beta=1}^i\quad\text{and}\quad\Delta_i^{(l)}=\det\mathcal{H}_i^{(l)},
\]
where $\Delta_0^l=1$ for any $l\in\Z$. 

\begin{thm}[{\cite[Theorem 4.1]{GSV4}}]
If $X$ possesses factorization \eqref{eq: Coxeter factorization}, then
\begin{equation*}\label{eq: c_i and d_i} 
\begin{aligned}
d_i&=\frac{\Delta_i^{(\varkappa_i+1)}\Delta_{i-1}^{(\varkappa_{i-1})}}{\Delta_i^{(\varkappa_i)}\Delta_{i-1}^{(\varkappa_{i-1}+1)}},\\
c_i^+c_i^-&=\frac{\Delta_{i-1}^{(\varkappa_{i-1})}\Delta_{i+1}^{(\varkappa_{i+1})}}{\left(\Delta_i^{(\varkappa_i+1)}\right)^2}\left(\frac{\Delta_{i+1}^{(\varkappa_{i+1}+1)}}{\Delta_{i+1}^{(\varkappa_{i+1})}}\right)^{\varepsilon_{i+1}}
\left(\frac{\Delta_{i-1}^{\varkappa_{i+1}+1}}{\Delta_{i-1}^{(\varkappa_{i-1})}}\right)^{2-\varepsilon_i}
\end{aligned}
\end{equation*}
for any $i\in[1,n]$.
\end{thm}

\begin{cor}
For a pair $(u,v)$ of Coxeter elements in $S_n$. The map $m_{u,v}:G^{u,v}/H\to\mathcal{W}_n$ is invertible. 
\end{cor}

\begin{rem}
{\em
The standard Poisson--Lie bracket on $\GL_n$ induces the Poisson bracket $\pb_1$ \eqref{eq: reduced F-G} on the Weyl functions. A proof using the weighted network approach can be found in \cite[Proposition 5.3]{GSV4}.
}
\end{rem}

Following \cite{GSV4}, one considers the parameters $c_i=c_i^+c_i^-,\,d_i$ that correspond to factorization \eqref{eq: Coxeter factorization} of a generic element in $G^{u,v}$ as coordinates on the open dense set in $G^{u,v}/H$. 

\begin{prop}[{\cite[Lemma 6.1]{GSV4}}]\label{prop: quadratic epsilon} 
The standard Poisson--Lie structure on $\GL_n$ induces the following Poisson brackets for variables $c_i,d_i$
\begin{equation}\label{eq: quadratic epsilon} 
\lbrace c_i,c_{i+1}\rbrace=(\varepsilon_{i+1}-1)c_ic_{i+1},\quad\lbrace c_i,d_i\rbrace=-c_id_i,\quad\lbrace c_i,d_{i+1}\rbrace=c_id_{i+1},
\end{equation}
and the rest of the brackets are zero.
\end{prop} 
The Hamiltonians $F_k(X)=\frac{1}{k}\tr(X^k)$, due to the invariance under conjugation by elements of $H$,  when restricted to a Coxeter double Bruhat cell $G^{u,v}$ can be expressed as functions of $c_i,d_i$. Thus, the functions $F_k$ ($k=1,\ldots,n-1$) serve as Hamiltonians for Coxeter--Toda flows on $G^{u,v}/H$. Therefore, the Toda hierarchy defines a completely integrable system on the symplectic leaves (level sets of the determinant) in $G^{u,v}/H$, i.e. Coxeter-Toda lattices are completely integrable systems.

In the case $k=1$, the Hamiltonian $F_1$ has the form
\begin{equation}\label{eq: F_1} 
F_1=\sum_{i=1}^n(d_i+c_{i-1}d_{i-1})+\sum_{i=3}^n\sum_{r=1}^{i-2}d_rc_r\prod_{j=r+1}^{i-1}c_j\varepsilon_j^-\varepsilon_j^+.
\end{equation}

Using \eqref{eq: F_1}, the Hamiltonian equations of the Coxeter--Toda flow generated by $F_1$ and \eqref{eq: quadratic epsilon} on $G^{u,v}/H$ are given by
\begin{equation}\label{eq: Coxeter-Toda flow} 
\begin{aligned}
\dot c_p=\lbrace c_p,F_1\rbrace&=c_p\left[d_{p+1}-d_p+(1-\varepsilon_p)c_{p-1}d_{p-1}-c_pd_p+\varepsilon_{p+1}c_{p+1}d_{p+1}\right]+\sum_{i=3}^n\sum_{r=1}^{i-2}\prod_{j=r+1}^{i-1}\varepsilon_j^-\varepsilon_j^+\\
&\mathrel{\phantom{=}}\times\left[(\delta_{p+1,r}-\delta_{p,r})c_pd_r\prod_{j=r}^{i-1}c_j+d_r\sum_{j=r}^{i-1}\left((\varepsilon_j-1)\delta_{j,p+1}+(1-\varepsilon_p)\delta_{j,p-1}\right)c_jc_p\prod_{l=r\atop l\neq j}^{i-1}c_l\right]\\
\dot d_p=\lbrace d_p,F_1\rbrace&=d_p\left[d_pc_p-d_{p-1}c_{p-1}\right]+\sum_{i=3}^n\sum_{r=1}^{i-2}d_r\prod_{j=r+1}^{i-1}\varepsilon_j^-\varepsilon_j^+\sum_{j=r}^{i-1}(\delta_{j,p}-\delta_{j,p-1})c_jd_p\prod_{l=r\atop l\neq j}^{i-1}c_l
\end{aligned}
\end{equation}

\begin{exa}\label{example 1} 
{\em
\begin{enumerate}
\item[(i)] Let $u^{-1}=v=s_{n-1}\cdots s_1$. Then $\varepsilon=(2,0,\ldots,0)$ and the Hamiltonian equations \eqref{eq: Coxeter-Toda flow} become 
\begin{equation}
\begin{aligned}
\dot c_i&=\lbrace c_i,F_1\rbrace=c_i(d_{i+1}-d_i+c_{i-1}d_{i-1}-c_id_i),\\
\dot d_i&=\lbrace d_i,F_1\rbrace=d_i(c_id_i-c_{i-1}d_{i-1}).
\end{aligned}
\end{equation}
Using the change of variables, $a_i=c_id_i^2,\,b_i=d_i+c_{i-1}d_{i-1}$, the system becomes the Toda lattice \eqref{eq: Toda lattice}.
\item[(ii)] Let $u=v=s_{n-1}\cdots s_1$. Then $\varepsilon=(2,1,\ldots,1,0)$ and the Hamiltonian equations \eqref{eq: Coxeter-Toda flow} become
\begin{align*}
\dot c_i=c_i(d_{i+1}-d_i+c_{i+1}d_{i+1}-c_id_i),\qquad\dot d_i=d_i(c_id_i-c_{i-1}d_{i-1}).
\end{align*}
After the change of variables $\tilde{c}_i=c_id_i$, this system becomes the {\it relativistic Toda lattice}:
\begin{equation}\label{eq: relativistic toda lattice} 
 \dot{\tilde{c}}_i=\tilde{c}_i(d_{i+1}-d_i+\tilde{c}_{i+1}-\tilde{c}_{i-1}),\qquad\dot d_i=d_i(\tilde{c}_i-\tilde{c}_{i-1}).
\end{equation}
\end{enumerate}
}
\end{exa}

It is well known that the Toda lattice \eqref{eq: Toda lattice} and the relativistic Toda lattice \eqref{eq: relativistic toda lattice} are completely integrable bi-Hamiltonian systems \cite{Su}. Moreover, one can use the family of compatible Faybusovich--Gekhtman Poisson brackets and the fact that the Moser map is an invertible map to guarantee the multi-Hamiltonian nature of Coxeter--Toda lattices. This fact was exposed in \cite{FG2} for {\it elementary Toda lattices}, which are particular Coxeter--Toda lattices $G^{u,v}/H$ by taking $v=s_{n-1}\cdots s_1$ and an arbitrary Coxeter element $u$.

\subsection{Generalized B\"acklund--Darboux transformations}

We now review the {\it generalized B\"acklund--Darboux transformation}. It is a birational automorphism between phase spaces of Coxeter--Toda lattices which preserve the corresponding Coxeter--Toda flows. The cluster algebra interpretation of these automorphisms was exposed in \cite{GSV4}. 

For the purpose of this note, we avoid the cluster algebra approach and simply we refer to \cite{GSV4} for the interested reader. However, based in \cite{GSV4}, we define these transformations in a more axiomatic way as a finite composition of some elementary transformations.

\begin{table}[h]
\begin{center}
\begin{tabular}{|c|c|c|}
\hline
$\varepsilon$ &  \text{Mutation} & $\varepsilon'$ \\
\hline
$\varepsilon_i=2,\varepsilon_{i+1}=0$ & $(0,i)$ & $\varepsilon'_i=1,\varepsilon'_{i+1}=1$\\
\hline
$\varepsilon_i=1,\varepsilon_{i+1}=0$ & $(0,i)$ & $\varepsilon'_i=0,\varepsilon'_{i+1}=1$\\
\hline
$\varepsilon_{n-1}=0$ & $(1,n-1)$ & $\varepsilon'_{n-1}=1$\\
\hline
$\varepsilon_{n-1}=1$ & $(1,n-1)$ & $\varepsilon'_{n-1}=2$\\
\hline
\end{tabular}
\end{center}
\caption{Mutations}\label{table 2} 
\end{table}

Consider two pairs $(u,v)$ and $(u',v')$ of Coxeter elements such that the entries of their corresponding $n$-tuples $\varepsilon=(\varepsilon_i)_{i=1}^n$ and $\varepsilon'=(\varepsilon'_i)_{i=1}^n$ satisfy one of the possible situations listed in Table \ref{table 2}. Here the entries of the $n$-tuples $\varepsilon$ and $\varepsilon'$  are all equal except only in the entries  specified in Table \ref{table 2}. In the first two rows $i$ is assumed to be less than $n-1$. 

We denote by $(s,i),\,s=0,1,\,i\in[1,n-1],$ as is indicated in Table \ref{table 2}, the transformation, called {\it mutation}, that takes the $n$-tuple $\varepsilon$ to the $n$-tuple $\varepsilon'$. If $(s,i)$, $s=0,1,\,i\in[1,n-1]$, denotes the mutation of  $n$-tuple $\varepsilon$ to  the $n$-tuple $\varepsilon'$ specified in Table \ref{table 2} then we denote by $(1-s,i),\,s=0,1$, the inverse mutation from the $n$-tuple $\varepsilon'$ to the $n$-tuple $\varepsilon$.

\begin{table}[h]
\begin{center}
\begin{tabular}{|c|c|c|c|}
\hline
$\varepsilon$ & $\varepsilon'$ &\text{Transformation} & \text{Inverse}\\
\hline
& & $c'_{i-1}=c_{i-1}(1+c_i)$ & $c_{i-1}=\frac{c'_{i-1}d'_{i+1}}{d'_{i+1}+c'_id'_i}$\\
$\varepsilon_i=2$ & $\varepsilon'_i=1$ & $c'_i=\frac{c_id_{i+1}}{d_i(1+c_i)^2}$ & $c_i=\frac{c'_id'_i}{d'_{i+1}}$\\
& & $c'_{i+1}=c_{i+1}(1+c_i)$ & $c_{i+1}=\frac{c'_{i+1}d'_{i+1}}{d'_{i+1}+c'_id'_i}$\\
$\varepsilon_{i+1}=0$ & $\varepsilon'_{i+1}=1$ & $d'_i=d_i(1+c_i)$ & $d_i=\frac{d'_id'_{i+1}}{d'_{i+1}+d'_ic'_i}$\\
&  & $d'_{i+1}=\frac{d_{i+1}}{1+c_i}$ & $d_{i+1}=d'_{i+1}+d'_ic'_i$\\
\hline
 & & $c'_i=\frac{c_id_{i+1}}{d_i(1+c_i)^2}$ & $c_i=\frac{c'_id'_i}{d'_{i+1}}$\\
$\varepsilon_i=1$ & $\varepsilon'_i=0$ & $c'_{i+1}=c_{i+1}(1+c_i)$ & $c_{i+1}=\frac{c'_{i+1}d'_{i+1}}{d'_{i+1}+c'_id'_i}$\\
$\varepsilon_{i+1}=0$ & $\varepsilon'_{i+1}=1$ & $d'_i=d_i(1+c_i)$ & $d_i=\frac{d'_id'_{i+1}}{d'_{i+1}+c'_id'_i}$\\
& &  $d'_{i+1}=\frac{d_{i+1}}{1+c_i}$ & $d_{i+1}=d'_{i+1}+c'_id'_i$\\
\hline
& & $c'_{n-1}=\frac{c_{n-1}d_{n-1}}{d_n}$ & $c_{n-1}=\frac{c'_{n-1}d'_n}{d'_{n-1}(1+c'_{n-1})^2}$\\
$\varepsilon_{n-1}=0$ & $\varepsilon'_{n-1}=1$ & $d'_{n-1}=\frac{d_nd_{n-1}}{d_n+c_{n-1}d_{n-1}}$ & $d_{n-1}=d'_{n-1}(1+c'_{n-1})$ \\
 & & $d'_n=d_n+c_{n-1}d_{n-1}$ & $d_n=\frac{d'_n}{1+c'_{n-1}}$\\
\hline
& & $c'_{n-2}=\frac{c_{n-2}d_n}{d_n+c_{n-1}d_{n-1}}$ & $c_{n-2}=c'_{n-2}(1+c'_{n-1})$ \\
$\varepsilon_{n-1}=1$ & $\varepsilon'_{n-1}=2$ & $c'_{n-1}=\frac{c_{n-1}d_{n-1}}{d_n}$ & $c_{n-1}=\frac{c'_{n-1}d'_n}{d'_{n-1}(1+c'_{n-1})^2}$\\
& & $d'_{n-1}=\frac{d_nd_{n-1}}{d_n+c_{n-1}d_{n-1}}$ & $d_{n-1}=d'_{n-1}(1+c'_{n-1})$\\
& & $d'_n=d_n+c_{n-1}d_{n-1}$ & $d_n=\frac{d'_n}{1+c'_{n-1}}$ \\
\hline
\end{tabular}
\end{center}
\caption{Elementary transformations}\label{table 3} 
\end{table}

There exists birational transformations between $G^{u,v}/H$ and $G^{u',v'}/H$, called {\it elementary transformations}, corresponding to the mutations of Table \ref{table 2}. These transformations are listed in Table \ref{table 3}.

Now, let us fix two arbirary pairs, $(u,v)$ and $(u',v')$ of Coxeter elements and let $\varepsilon=(\varepsilon_i)_{i=1}^n,\,\varepsilon'=(\varepsilon'_i)_{i=1}^n$ be their corresponding $n$-tuples. It was noted in \cite{GSV4} that we can transform the $n$-tuple $\varepsilon$ to the $n$-tuple $\varepsilon'$ via a finite sequence of mutations listed in Table \ref{table 2}.

\begin{exa}
{\em
Let $n=4$. The diagram below show the mutations of Table \ref{table 2} among all possible values of $\varepsilon$ corresponding to a pair of Coxeter elements $(u,v)$ of $S_4$.
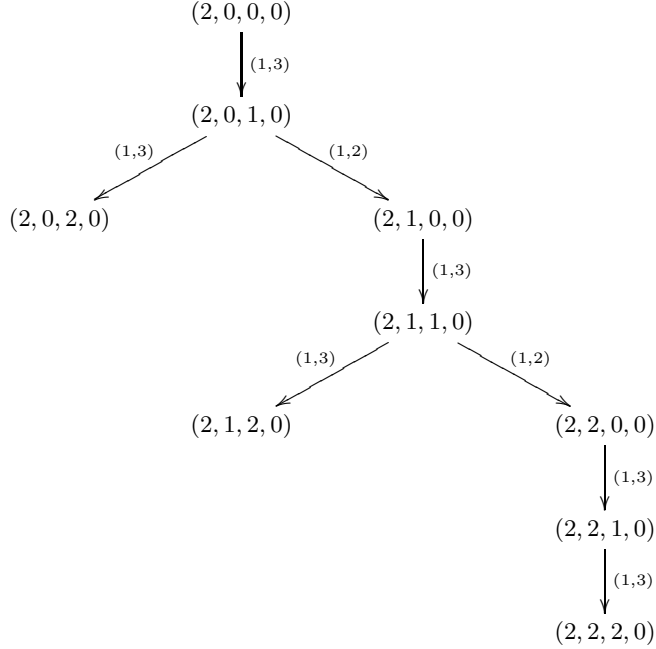
\begin{displaymath}
    \xymatrix{
         & (2,0,0,0)\ar[d]^{(1,3)}  & & \\
         & (2,0,1,0)\ar[ld]_{(1,3)}\ar[dr]^{(1,2)} & &\\
         (2,0,2,0) & & (2,1,0,0)\ar[d]^{(1,3)} &\\
         & & (2,1,1,0)\ar[ld]_{(1,3)}\ar[dr]^{(1,2)} &\\
         & (2,1,2,0) & & (2,2,0,0)\ar[d]^{(1,3)}\\
         & & & (2,2,1,0)\ar[d]^{(1,3)}\\
         & & & (2,2,2,0)
         }
\end{displaymath}
}
\end{exa}

The generalized B\"acklund--Darboux transformation $\sigma_{u,v}^{u',v'}:G^{u,v}/H\to G^{u',v'}/H$ is the corresponding finite composition of the elementary transformations, listed in Table \ref{table 3}, associated to the finite sequence of mutations that takes the $n$-tuple $\varepsilon$ into the $n$-tuple $\varepsilon'$.


\section{Main results}

Despite the fact that Coxeter--Toda lattices are completely integrable multi-Hamiltonian systems, no other Poisson bracket in $c_i,d_i$ coordinates and compatible with \eqref{eq: quadratic epsilon} is known in the literature. 

\subsection{A non-local Poisson bracket}


We consider the following non-local bracket $\pb_{nl}$

\begin{align}
\lbrace d_i,d_{i+k}\rbrace_{nl}&=
\begin{cases}
c_id_i & \text{ if } k=1,\\
\displaystyle{d_i\prod_{j=i}^{i+k-1}c_j} & \text{ if } k>1,\,\epsilon_{i+1}=\ldots=\epsilon_{i+k-1}=0.
\end{cases}\label{eq: rat 1} 
\\
\lbrace c_i,d_{i+k}\rbrace_{nl}&=
\begin{cases}
c_i & \text{ if } k=0,\\
-c_i(c_i+1) & \text{ if } k=1,\\
-\displaystyle{(c_i+1)\prod_{j=i}^{i+k-1}c_j} & \text{ if } k>1,\,\epsilon_{i+1}=\ldots=\epsilon_{i+k-1}=0,\\
\displaystyle{\prod_{j=i}^{i+k-1}c_j} & \text{ if } k>1,\,\epsilon_{i+1}=2,\epsilon_{i+2}=\ldots=\epsilon_{i+k-1}=0.
\end{cases}\label{eq: rat 2} 
\\
\lbrace c_{i+k},d_i\rbrace_{nl}=&
\begin{cases}
\displaystyle{(2-\epsilon_{i+1})c_ic_{i+1}\frac{d_i}{d_{i+1}}} & \text{ if } k=1,\\
\displaystyle{(2-\epsilon_{i+k})\frac{d_i}{d_{i+k}}\prod_{j=i}^{i+k}c_j} & \text{ if } k>1,\,\epsilon_{i+1}=\ldots=\epsilon_{i+k-1}=0.
\end{cases}\label{eq: rat 3} 
\end{align}
\begin{align}
\lbrace c_i,c_{i+k}\rbrace_{nl}=&
\begin{cases}
\displaystyle{(2-\epsilon_{i+1})\frac{c_i+1}{d_{i+1}}c_ic_{i+1}} & \text{ if } k=1,\\
\displaystyle{(2-\epsilon_{i+k})\frac{c_i+1}{d_{i+k}}\prod_{j=i}^{i+k}}c_j & \text{ if } k>1,\,\epsilon_{i+1}=\ldots=\epsilon_{i+k-1}=0,\\
\displaystyle{(\epsilon_{i+k}-2)\frac{\prod_{j=i}^{i+k}c_j}{d_{i+k}}} & \text{ if } k>1,\,\epsilon_{i+1}=2,\epsilon_{i+2}=\ldots=\epsilon_{i+k-1}=0.
\end{cases}\label{eq: rat 4} 
\end{align}
and the rest of the brackets are zero.

\begin{rem}
{\em
The Poisson bracket \eqref{eq: rat 1}$-$\eqref{eq: rat 4} is a consequence of a study of the Moser map $m_{u,v}:G^{u,v}\to\left(\mathcal{W}_n,\pb_0\right)$ for small values of $n$, where $\pb_0$ is the Atiyah-Hitchin bracket.
}
\end{rem}

\begin{exa}
{\em
The case $n=3$ is shown in the Table \ref{table 1}. 
\begin{table}[h]
\begin{center}
\begin{tabular}{|c|c|c|c|}
\hline
& $u^{-1}=s_1s_2, v=s_1s_2$ & $u^{-1}=s_2s_1, v=s_1s_2$ & $u^{-1}=s_2s_1, v=s_2s_1$\\
\hline
$\lbrace c_1,c_2\rbrace_{nl}$ & 0 & $\frac{c_1c_2(c_1+1)}{d_2}$ & $\frac{2c_1c_2(c_1+1)}{d_2}$\\
\hline
$\lbrace c_1,d_1\rbrace_{nl}$ & $c_1$ & $c_1$ & $c_1$\\
\hline
$\lbrace c_1,d_2\rbrace_{nl}$ & $-c_1(c_1+1)$ & $-c_1(c_1+1)$ & $-c_1(c_1+1)$\\
\hline
$\lbrace c_1,d_3\rbrace_{nl}$ & $c_1c_2$ & $0$ & $-c_1c_2(c_1+1)$\\
\hline
$\lbrace c_2,d_1\rbrace_{nl}$ & 0 & $\frac{c_1c_2d_1}{d_2}$ & $\frac{2c_1c_2d_1}{d_2}$\\
\hline
$\lbrace c_2,d_2\rbrace_{nl}$ & $c_2$ & $c_2$ & $c_2$\\
\hline
$\lbrace c_2,d_3\rbrace_{nl}$ & $-c_2(c_2+1)$ & $-c_2(c_2+1)$ & $-c_2(c_2+1)$\\
\hline
$\lbrace d_1,d_2\rbrace_{nl}$ & $c_1d_1$ & $c_1d_1$ & $c_1d_1$\\
\hline
$\lbrace d_1,d_3\rbrace_{nl}$ & 0 & 0 & $c_1c_2d_1$\\
\hline
$\lbrace d_2,d_3\rbrace_{nl}$ & $c_2d_2$ & $c_2d_2$ & $c_2d_2$\\
\hline
$\epsilon_1$ & 2 & 2 & 2 \\
\hline
$\epsilon_2$ & 2 & 1 & 0 \\
\hline
$\epsilon_3$ & 0 & 0 & 0 \\
\hline
\end{tabular}
\end{center}
\caption{Non-local bracket for $n=3$}\label{table 1}
\end{table}
}
\end{exa}

\begin{thm}\label{thm: 1} 
The bracket given by relations \eqref{eq: rat 1}$-$\eqref{eq: rat 4} is a Poisson bracket.
\end{thm}

\begin{proof}
Letting
\begin{align}
\varepsilon_{i+1}&=\ldots=\varepsilon_{i+m-1}=0,\tag{A1} \\
\varepsilon_{i+1}&=2,\varepsilon_{i+2}=0,\ldots=\varepsilon_{i+m-1}=0,\tag{A2}\\
\varepsilon_{i+m+1}&=\cdots=\varepsilon_{i+m+n-1}=0,\tag{B1}\\
\varepsilon_{i+m+1}&=2,\varepsilon_{i+m+2}=\cdots=\varepsilon_{i+m+n-1}=0,\tag{B2}.
\end{align} 
The Jacobi property for the bracket \eqref{eq: rat 1}--\eqref{eq: rat 4} follows from verifying it for 8 possible triples of functions  $d_i, d_{i+m},d_{i+m+n}$; $d_i,d_{i+m},c_{i+m+n}$; $d_i,c_{i+m},d_{i+m+n}$, etc.\\
Each case is verified separately depending on whether conditions (A1), (A2), (B1) and (B2) are satisfied. We will only provide a justification for the equation
\begin{equation}\label{eq: case1} 
\lbrace d_i,\lbrace d_{i+m},d_{i+m+n}\rbrace_{nl}\rbrace_{nl}+\lbrace d_{i+m},\lbrace d_{i+m+n},d_i\rbrace_{nl}\rbrace_{nl}+\lbrace d_{i+m+n},\lbrace d_i,d_{i+m}\rbrace_{nl}\rbrace_{nl}=0.
\end{equation}
The remaining seven cases can be treated similarly (see \cite{Eber}).

To prove equation \eqref{eq: case1} we will use only (A1) and (B1) considering the following three cases

{\bf Case 1:} $m>1, n>1$, 
\begin{enumerate}
\item[(a)] If condition (A1) is false we have two possibilities

\begin{enumerate}
\item[(i)] Condition (B1) is false then \eqref{eq: case1} is zero.
\item[(ii)] Condition (B1) is true then \eqref{eq: case1} is equal to
\[
\lbrace d_i,d_{i+m}\prod_{j=i+m}^{i+m+n-1}c_j\rbrace_{nl}=\lbrace d_i,d_{i+m}\rbrace_{nl}\prod_{j=i+m}^{i+m+n-1}c_j+d_{i+m}\lbrace d_i,\prod_{j=i+m}^{i+m+n-1}c_j\rbrace_{nl}=0.
\]
\end{enumerate}

\item[(b)] If condition (A1) is true 

\begin{enumerate}
\item[(i)] If condition (B1) is false then \eqref{eq: case1} becomes
\[
\lbrace d_{i+m+n},d_i\prod_{j=i}^{i+m-1}c_j\rbrace_{nl}=d_i\lbrace d_{i+m+n},\prod_{j=i}^{i+m-1}c_j\rbrace_{nl}=0.
\]
The last equality follows from the bracket $\lbrace d_{i+m+n},c_l\rbrace_{nl}=0,\,l=i,\ldots,i+m-1$, since (B1) is false.
\item[(ii)] If condition (B1) is true and $\varepsilon_{i+m}\neq 0$ then \eqref{eq: case1} becomes
\begin{equation}\label{eq: case1-a} 
\lbrace d_i,d_{i+m}\rbrace_{nl}\prod_{j=i+m}^{i+m+n-1}c_j
+d_{i+m}\lbrace d_i,\prod_{j=i+m}^{i+m+n-1}c_j\rbrace_{nl} +d_i\lbrace d_{i+m+n},\prod_{j=i}^{i+m-1}c_j\rbrace_{nl}
\end{equation}
Now, since $\varepsilon_{i+m}\neq 0$, one has 
\begin{align*}
\lbrace d_i,c_l\rbrace_{nl}&=0,\quad\text{for } l=i+m+1,\ldots,i+m+n-1\\
\lbrace d_{i+m+n},c_l\rbrace_{nl}&=0,\quad\text{for } l=i,\ldots,i+m-2.
\end{align*}
Therefore, \eqref{eq: case1-a} is equal to
\[
d_i\prod_{j=i}^{i+m+n-1}c_j+(\varepsilon_{i+m}-2)d_i\prod_{j=i}^{i+m+n-1}c_j+\lbrace d_{i+m+n},c_{i+m-1}\rbrace_{nl}d_i\prod_{j=i}^{i+m-2}c_j 
\]
which becomes equal to zero for $\varepsilon_{i+m}=1,2$ since
\[
\lbrace d_{i+m+n},c_{i+m-1}\rbrace_{nl}=
\begin{cases}
-\prod_{j=i+m-1}^{i+m+n-1}c_j & \text{ if } \varepsilon_{i+m}=2\\
0&\text{ if } \varepsilon_{i+m}=1.
\end{cases}
\]
\item[(iii)] If condition (B1) is true and $\varepsilon_{i+m}=0$ then \eqref{eq: case1} is equal to
\begin{equation}
\begin{aligned}\label{eq: case1-b} 
d_i\prod_{j=i}^{i+m+n-1}c_j+d_{i+m}\lbrace d_i,\prod_{j=i+m}^{i+m+n-1}c_j\rbrace -d_i\prod_{j=i+m}^{i+m+n-1}c_j\lbrace d_{i+m},\prod_{j=i}^{i+m-1}c_j\rbrace\\
-d_i\prod_{j=i}^{i+m-1}c_j\lbrace d_{i+m},\prod_{j=i+m}^{i+m+n-1}c_j\rbrace+d_i\lbrace d_{i+m+n},\prod_{j=i}^{i+m-1}c_j\rbrace
\end{aligned}
\end{equation}
Using the following identities
\begin{align*}
\lbrace d_i,\prod_{j=i+m}^{i+m+n-1}c_j\rbrace_{nl} &= -2d_i\prod_{j=i+m}^{i+m+n-1}c_j\sum_{r=i+m}^{i+m+n-1}\frac{1}{d_r}\prod_{j=i}^{r-1}c_j,\\
\prod_{j=i+m}^{i+m+n-1}c_j\lbrace d_{i+m},\prod_{j=i}^{i+m-1}c_j\rbrace_{nl}&= \prod_{j=i}^{i+m+n-1}c_j\left(\sum_{r=i}^{i+m-2}(c_r+1)\prod_{j=r+1}^{i+m-1}c_j+c_{i+m-1}+1\right),\\
\prod_{j=i}^{i+m-1}c_j\lbrace d_{i+m},\prod_{j=i+m}^{i+m+n-1}c_j\rbrace_{nl}&= -\prod_{j=i}^{i+m+n-1}c_j\left(2d_{i+m}\sum_{r=i+m+2}^{i+m+n-1}\frac{1}{d_r}\prod_{j=i+m}^{r-1}c_j+2c_{i+m}\frac{d_{i+m}}{d_{i+m+1}}+1\right),\\
\lbrace d_{i+m+n},\prod_{j=i}^{i+m-1}c_j\rbrace_{nl}&=\prod_{j=i}^{i+m-1}c_j\sum_{r=i}^{i+m-1}(c_r+1)\prod_{j=r+1}^{i+m+n-1}c_j,
\end{align*}
one can show that \eqref{eq: case1-b} is equal to zero.
\end{enumerate}
\end{enumerate}

{\bf Case 2:} $m=1,\,n>1$. In this case we let 
\begin{align*}
\varepsilon_{i+2}&=\cdots=\varepsilon_{i+n}=0.\tag{B3}
\end{align*} 
\begin{enumerate}
\item[(a)] If condition (B3) is false then $\lbrace d_j,d_{i+1+n}\rbrace_{nl}=0,\,j=i,i+1$ and $\lbrace c_i,d_{i+1+n}\rbrace_{nl}=0$ . Therefore, \eqref{eq: case1} is equal to zero.
\item[(b)] If condition (B3) is true then \eqref{eq: case1} is equal to
\begin{multline}\label{eq: case1-2} 
c_id_i\prod_{j=i+1}^{i+n}c_j+d_{i+1}\sum_{k=i+1}^{i+n}\prod_{r=i+1\atop r\neq k}^{i+n}c_r\lbrace d_i,c_k\rbrace_{nl}+\lbrace d_{i+1},\lbrace d_{i+1+n}, d_i\rbrace_{nl}\rbrace_{nl}\\
+c_i\lbrace d_{i+1+n},d_i\rbrace_{nl}+d_i\lbrace d_{i+1+n},c_i\rbrace_{nl}
\end{multline}
\begin{enumerate}
\item[(i)] If $\varepsilon_{i+1}\neq 0$ then $\lbrace d_i,d_{i+1+n}\rbrace_{nl}=0$ and $\lbrace d_i,c_k\rbrace_{nl}=0,\,k=i+2,\ldots,i+n$. The equation \eqref{eq: case1-2} becomes
\[
d_i\prod_{j=i}^{i+n}c_j+d_{i+1}\lbrace d_i,c_{i+1}\rbrace_{nl}\prod_{r=i+2}^{i+n}c_r+d_i\lbrace d_{i+1+n},c_i\rbrace_{nl},
\]
which is zero since
\[
\lbrace d_i,c_{i+1}\rbrace_{nl}=
\begin{cases}
0&\text{ if } \varepsilon_{i+1}=2\\
-\frac{c_ic_{i+1}d_i}{d_{i+1}} &\text{ if } \varepsilon_{i+1}=1
\end{cases},\qquad
\lbrace d_{i+1+n},c_i\rbrace_{nl}=
\begin{cases}
-\prod_{j=i}^{i+n}c_j&\text{ if } \varepsilon_{i+1}=2\\
0 &\text{ if } \varepsilon_{i+1}=1
\end{cases}.
\]
\item[(ii)] If $\varepsilon_{i+1}= 0$ then \eqref{eq: case1-2}, using \eqref{eq: rat 1} and \eqref{eq: rat 2}, becomes
\begin{equation}\label{eq: case1-2a} 
2d_i\prod_{j=i}^{i+n}c_j+d_{i+1}\sum_{k=i+1}^{i+n}\prod_{r=i+1\atop r\neq k}^{i+n}c_r\lbrace d_i,c_k\rbrace_{nl}+c_id_i\prod_{j=i}^{i+n}c_j-d_i\sum_{k=i}^{i+n}\prod_{r=i\atop r\neq k}^{i+n}c_r\lbrace d_{i+1},c_k\rbrace_{nl}.
\end{equation}
Using the fact that $\varepsilon_{i+1}=0$, we have
\begin{align*}
\sum_{k=i+1}^{i+n}\prod_{r=i+1\atop r\neq k}^{i+n}c_r\lbrace d_i,c_k\rbrace_{nl}&=-\frac{d_i}{d_{i+1}}\prod_{r=i}^{i+n}c_r+\sum_{k=i+2}^{i+n}\prod_{r=i+1\atop r\neq k}^{i+n}c_r\lbrace d_i,c_k\rbrace_{nl},\\
\sum_{k=i}^{i+n}\prod_{r=i\atop r\neq k}^{i+n}c_r\lbrace d_{i+1},c_k\rbrace_{nl}&=-(c_i+1)\prod_{j=i}^{i+n}c_j+\prod_{j=i}^{i+n}c_j-\sum_{k=i+2}^{i+n}\prod_{r=i\atop i+n}^{in+n}c_r\lbrace d_{i+1},c_k\rbrace_{nl},
\end{align*}
which lead to \eqref{eq: case1-2a} be equal to zero.
\end{enumerate}
\end{enumerate}

{\bf Case 3:} $m=n=1$. In this case the equation \eqref{eq: case1} reduces to
\begin{equation}\label{eq: case1-3} 
d_{i+1}\lbrace d_i,c_{i+1}\rbrace_{nl}+c_ic_{i+1}d_i+\lbrace d_{i+1},\lbrace d_{i+2},d_i\rbrace_{nl}\rbrace_{nl}
+c_i\lbrace d_{i+2},d_i\rbrace_{nl}+d_i\lbrace d_{i+2},c_i\rbrace_{nl}
\end{equation}
\begin{enumerate}
\item[(a)] If $\varepsilon_{i+1}\neq 0$ then $\lbrace d_i,d_{i+2}\rbrace_{nl}=0$ and the equation \eqref{eq: case1-3} becomes
\[
(\varepsilon_{i+1}-1)c_ic_{i+1}d_i+d_i\begin{cases}-c_ic_{i+1}&\text{ if } \varepsilon_{i+1}=2\\
0 &\text{ if } \varepsilon_{i+1}=1\end{cases}
\]
which is zero for $\varepsilon_{i+1}=1,2$.
\item[(b)] If $\varepsilon_{i+1}= 0$ then, using \eqref{eq: rat 1}--\eqref{eq: rat 2}, the equation \eqref{eq: case1-3} is reduced to
\[
-\lbrace d_{i+1},c_ic_{i+1}d_{i+1}\rbrace_{nl}=c_ic_{i+1}\lbrace d_i,d_{i+1}\rbrace_{nl}+c_{i+1}d_i\lbrace c_i,d_{i+1}\rbrace_{nl}+c_i,d_i\lbrace c_{i+1},d_{i+1}\rbrace_{nl}=0
\]
\end{enumerate}
\end{proof}

\begin{exa}\label{example 2} 
{\em
Let $u=v=s_{n-1}\cdots s_1$. Then $\varepsilon=(2,1,\ldots,1,0)$ and the Poisson bracket \eqref{eq: rat 1}--\eqref{eq: rat 4} becomes 
\begin{equation}\label{eq: relativistic toda} 
\begin{aligned}
\lbrace d_i,d_{i+1}\rbrace_{nl}&=d_ic_i,&\lbrace c_i,c_{i+1}\rbrace_{nl}&=\frac{(c_i+1)c_ic_{i+1}}{d_{i+1}},\\
\lbrace c_{i+1},d_i\rbrace_{nl} &=\frac{d_ic_ic_{i+1}}{d_{i+1}},& \lbrace c_i,d_i\rbrace_{nl}&=c_i,\\
\lbrace c_i,d_{i+1}\rbrace_{nl}&=-c_i(c_i+1).
\end{aligned}
\end{equation}
Using the change of variables $\tilde{c}_i=c_id_i$, the Poisson bracket \eqref{eq: relativistic toda} becomes
\[
\lbrace \tilde{c}_i,d_i\rbrace_{nl}=\tilde{c}_i=\lbrace d_i,d_{i+1}\rbrace_{nl}=-\lbrace\tilde{c}_i,d_{i+1}\rbrace_{nl},
\]
which is a linear Poisson bracket for the relativistic Toda lattice \eqref{eq: relativistic toda lattice} (see \cite{Su}).
}
\end{exa}

Consider now $u^{-1}=v=s_{n-1}\ldots s_1$. Then $\varepsilon=(2,0,\ldots,0)$ and, due to Proposition \ref{prop: factorization}, generic elements $L\in G^{u,v}/H$ are described as the set of Jacobi matrices 
\begin{equation}\label{eq: Coxeter tridiagonal} 
L=\begin{pmatrix}
d_1 & 1  &  0 & \ldots & 0\\
c_1d_1^2 & d_2+c_1d_1 & 1 & \ldots &0\\
\vdots & \ddots & \ddots &\ddots & \vdots \\
0 & \cdots & * & * & 1\\
0 & \ldots & 0 & c_{n-1}d_{n-1}^2 & d_n+c_{n-1}d_{n-1}
\end{pmatrix}.
\end{equation}
Moreover, the non-local Poisson bracket \eqref{eq: rat 1}--\eqref{eq: rat 4} becomes
\begin{equation}\label{eq: tridiagonal case} 
\begin{aligned}
\lbrace d_i,d_j\rbrace_{nl}&=d_i\prod_{l=i}^{j-1}c_l,&\lbrace c_i,c_j\rbrace_{nl}&=\frac{2(c_i+1)\prod_{l=i}^{j}c_l}{d_j}, &\lbrace c_i,d_i\rbrace_{nl} =c_i\\
\lbrace c_j,d_i\rbrace_{nl}&=\frac{2d_i\prod_{l=i}^{j}c_l}{d_j},& \lbrace c_i,d_j\rbrace_{nl}&=-(c_i+1)\prod_{l=i}^{j-1}c_l.\\
\end{aligned}
\end{equation}
By means of the change of variables $a_i=c_id_i^2,\, b_i=d_i+c_{i-1}d_{i-1}$, it is not hard to show that the Poisson bracket \eqref{eq: tridiagonal case}  becomes 
\begin{align}\label{eq: linear Toda bracket} 
\lbrace a_i,b_i\rbrace_{nl}=-\lbrace a_i,b_{i+1}\rbrace_{nl}=a_i.
\end{align}
The latter linear Poisson bracket and $F_2(a,b)=\sum_{i=1}^{n-1}a_i+\frac{1}{2}\sum_{i=1}^nb_i^2$ generate the non-periodic Toda lattice \eqref{eq: Toda lattice}. Using the new variables $a_i,\,b_i$, for a generic element $L\in G^{u,v}/H$ as in \eqref{eq: Coxeter tridiagonal}, we obtain the following Poisson relations
\begin{equation}\label{eq: equations flaschka} 
\nl{b_i,\Delta_{\left[i+1,n\right]}(\lambda)}=0,\qquad \nl{a_i,\Delta_{\left[i+1,n\right]}(\lambda)}=a_i\Delta_{\left[i+2,n\right]}(\lambda),\quad i=1,\ldots,n-1.
\end{equation}
In fact, since $\Delta_{\left[i+k,n\right]}(\lambda)$ involve only $a_j,\,b_j,\,j=i+k,\ldots,n$, and 
\[
\Delta_{\left[i+1,n\right]}(\lambda)=(\lambda-b_{i+1})\Delta_{\left[i+2,n\right]}(\lambda)-a_{i+1}\Delta_{\left[i+3,n\right]}(\lambda)
\]
we have
\begin{align*}
\nl{b_i,\Delta_{\left[i+1,n\right]}(\lambda)}&=\nl{b_i,(\lambda-b_{i+1})\Delta_{\left[i+2,n\right]}(\lambda)-a_{i+1}\Delta_{\left[i+3,n\right]}(\lambda)}\\
&=(\lambda-b_{i+1})\nl{b_i,\Delta_{\left[i+2,n\right]}(\lambda)}-\Delta_{\left[i+2,n\right]}(\lambda)\nl{b_i,b_{i+1}}\\
&\mathrel{\phantom{=}}-\Delta_{\left[i+3,n\right]}(\lambda)\nl{b_i,a_{i+1}}-a_{i+1}\nl{b_i,\Delta_{[i+3,n]}(\lambda)}\\
&=0,
\end{align*}
\begin{align*}
\nl{a_i,\Delta_{\left[i+1,n\right]}(\lambda)}&=\nl{a_i,(\lambda-b_{i+1})\Delta_{\left[i+2,n\right]}(\lambda)-a_{i+1}\Delta_{\left[i+3,n\right]}(\lambda)}\\
&=(\lambda-b_{i+1})\nl{a_i,\Delta_{\left[i+2,n\right]}(\lambda)}-\Delta_{\left[i+2,n\right]}(\lambda)\nl{a_i,b_{i+1}}\\
&\mathrel{\phantom{=}}-\Delta_{\left[i+3,n\right]}(\lambda)\nl{a_i,a_{i+1}}-a_{i+1}\nl{a_i,\Delta_{[i+3,n]}(\lambda)}\\
&=a_i\Delta_{[i+2,n]}(\lambda).
\end{align*}

Consider now the $2\times 2$ matrix
\[
L_i=\begin{pmatrix} b_i & 1\\ a_i & b_{i+1}\end{pmatrix}
\]
and its Weyl map
\[
m(\lambda)=m(\lambda;L_i)=\frac{\Delta_{[2,2]}(\lambda)}{\Delta_{[1,2]}(\lambda)}=\frac{\lambda-b_{i+1}}{\lambda^2-(b_i+b_{i+1})\lambda+b_ib_{i+1}-a_i}
\]
From the linear Poisson bracket \eqref{eq: linear Toda bracket}  follows that
\begin{align*}
-\nl{\Delta_{[1,2]}(\lambda),\Delta_{[2,2]}(\mu)}&=-\nl{\lambda^2-(b_i+b_{i+1})\lambda+b_ib_{i+1}-a_i,\mu-b_{i+1}}=-\nl{a_i,b_{i+1}}=a_i\\
&\stackrel{\eqref{eq: AH n=2}}{=}\frac{\Delta_{[1,2]}(\lambda)\Delta_{[1,1]}(\mu)-\Delta_{[1,2]}(\mu)\Delta_{[1,1]}(\lambda)}{\lambda-\mu}-\Delta_{[2,2]}(\lambda)\Delta_{[2,2]}(\mu),\\
\nl{\Delta_{[1,2]}(\lambda),\Delta_{[1,2]}(\mu)}&=\nl{\lambda^2-(b_i+b_{i+1})\lambda+b_ib_{i+1}-a_i,\mu^2-(b_i+b_{i+1})\mu+b_ib_{i+1}-a_i}\\
&=\nl{b_i+b_{i+1},a_i}\lambda+\nl{a_i,b_i+b_{i+1}}\mu=0,\\
\nl{\Delta_{[2,2]}(\lambda),\Delta_{[2,2]}(\mu)}&=0.
\end{align*}
Thus, in particular, the Moser map $m(\lambda):\left(G^{u,v}/H,-\pb_{nl}\right)\to\left(\mathcal{W}_2,\pb_0\right)$ is a Poisson map. 

Now, we consider the $3\times 3$ matrix
\[
L_i=\begin{pmatrix} b_i & 1 & 0\\ a_i & b_{i+1} & 1\\ 0 & a_{i+1} & b_{i+2}\end{pmatrix}
\]
and its Weyl function 
\[
m(\lambda)=m(\lambda;L_i)=\frac{\Delta_{[2,3]}(\lambda)}{\Delta_{[1,3]}(\lambda)}.
\]
Because of the discussion above for the $2\times 2$ matrix, the fact that $\Delta_{[1,3]}(\lambda)=(\lambda-b_i)\Delta_{[2,3]}(\lambda)-a_i\Delta_{[3,3]}(\lambda)$ and equation \eqref{eq: equations flaschka}, we obtain
\begin{align*}
-\nl{\Delta_{[1,3]}(\lambda),\Delta_{[2,3]}(\mu)}&=-\left[(\lambda-b_i)\nl{\Delta_{[2,3]}(\lambda),\Delta_{[2,3]}(\mu)}-\Delta_{[2,3]}(\lambda)\nl{b_i,\Delta_{[2,3]}(\mu)}\right.\\
&\mathrel{\phantom{=}}\left.-a_i\nl{\Delta_{[3,3]}(\lambda),\Delta_{[2,3]}(\mu)}-\Delta_{[3,3]}(\lambda)\nl{a_i,\Delta_{[2,3]}(\mu)}\right]\\
&=-a_i\nl{\Delta_{[2,3]}(\mu),\Delta_{[3,3]}(\lambda)}+a_i\Delta_{[3,3]}(\lambda)\Delta_{[3,3]}(\mu)\\
&=a_i\frac{\Delta_{[2,3]}(\mu)\Delta_{[3,3]}(\lambda)-\Delta_{[2,3]}(\lambda)\Delta_{[3,3]}(\mu)}{\mu-\lambda}.
\end{align*}
On the other hand, using the Atiyah--Hitchin bracket \eqref{eq: reduced F-G},
\begin{align*}
\left\{\Delta_{[1,3]}(\lambda),\Delta_{[2,3]}(\mu)\right\}_0&=\frac{\Delta_{[1,3]}(\lambda)\Delta_{[2,3]}(\mu)-\Delta_{[1,3]}(\mu)\Delta_{[2,3]}(\lambda)}{\lambda-\mu}-\Delta_{[2,3]}(\lambda)\Delta_{[2,3]}(\mu)\\
&=\frac{\left[(\lambda-b_i)\Delta_{[2,3]}(\lambda)-a_i\Delta_{[3,3]}(\lambda)\right]\Delta_{[2,3]}(\mu)-\left[(\mu-b_i)\Delta_{[2,3]}(\mu)-a_i\Delta_{[3,3]}(\mu)\right]\Delta_{[2,3]}(\lambda)}{\lambda-\mu}\\
&\mathrel{\phantom{=}}-\Delta_{[2,3]}(\lambda)\Delta_{[2,3]}(\mu)\\
&=a_i\frac{\Delta_{[2,3]}(\lambda)\Delta_{[3,3]}(\mu)-\Delta_{[2,3]}(\mu)\Delta_{[3,3]}(\lambda)}{\lambda-\mu}.
\end{align*}
Consequently,
\[
-\nl{\Delta_{[1,3]}(\lambda),\Delta_{[2,3]}(\mu)}=\left\{\Delta_{[1,3]}(\lambda),\Delta_{[2,3]}(\mu)\right\}_0.
\]
In a similar way, one can verify
\[
\nl{\Delta_{[1,3]}(\lambda),\Delta_{[1,3]}(\mu)}=\nl{\Delta_{[2,3]}(\lambda),\Delta_{[2,3]}(\mu)}=0.
\]
Therefore, in particular, if $u^{-1}=v=s_2s_1$ then the Moser map $m(\lambda):\left(G^{u,v}/H,-\pb_{nl}\right)\to\left(\mathcal{W}_3,\pb_0\right)$ is a Poisson map. Proceeding by induction on $n$, we obtain the next result. 
\begin{prop}\label{prop: 1} 
If $u^{-1}=v=s_{n-1}\cdots s_1\in S_n$ then the Moser map
\[
m=m_{u,v}\colon\left(G^{u,v}/H,-\pb_{nl}\right)\to\left(\mathcal{W}_n,\pb_0\right)
\]
is a Poisson map.
\end{prop}

To obtain the above result for any pair $(u,v)$ of Coxeter elements of $S_n$, we will need the next Theorem.

\begin{thm}\label{prop: non local poisson maps} 
The map $\sigma_{u,v}^{u',v'}:G^{u,v}/H\to G^{u',v'}/H$ is Poisson with respect to the Poisson structures \eqref{eq: rat 1}--\eqref{eq: rat 4} on $G^{u,v}/H$ and $G^{u',v'}/H$.
\end{thm}

\begin{cor}\label{cor: 1} 
For any pair of Coxeter elements $u,v$, the Moser map $m_{u,v}:\left(G^{u,v}/H,-\pb_{nl}\right)\to\left(\mathcal{W}_n,\pb_0\right)$ is a Poisson map.
\end{cor}

\begin{proof}
This is a consequence from Proposition \ref{prop: 1} and Theorem \ref{prop: non local poisson maps} and the fact that the generalized B\"acklund--Darboux map $\sigma_{u,v}^{u',v'}$ preserves the Weyl function.
\end{proof}

\begin{cor}\label{cor: 2} 
The non-local Poisson structure \eqref{eq: rat 1}--\eqref{eq: rat 4} and the quadratic Poisson bracket \eqref{eq: quadratic epsilon} are compatible.
\end{cor}
\begin{proof}
This is a consequence of  Propositions \ref{prop: compatible} and \ref{prop: quadratic epsilon}, Corollary \ref{cor: 2}, and the fact that the standard Poisson--Lie bracket $\pb_{\GL_n}$ induces $\pb_1$ on Weyl functions.
\end{proof}

\subsection{Discussion and Problems}

\begin{enumerate}
\item {\em Other compatible Poisson brackets}. A natural question is to compute explicit Poisson brackets on $c_i,d_i$ coordinates corresponding to the others members of the family of compatible Faybusovich--Gekhtman Poisson brackets for $k=2,\ldots,n-1$. 
\item{\em Poisson relations in factorization parameters}. In \cite[Equation 4.1]{KoZ}, a general formulae of Poisson relations, using factorization parameters $t_k,\,a^{\gamma}$, is induced from the standard Poisson--Lie structure on a simple Lie group $G$. The Poisson bracket \eqref{eq: rat 1}--\eqref{eq: rat 4} can suggest the existence of another Poisson bracket in factorization parameters $t_k,\,a^{\gamma}$ which will be compatible with the quadratic one presented in \cite{KoZ}.
\end{enumerate}

\appendix

\section{Proof of Theorem \ref{prop: non local poisson maps}}\label{Appendix A}

\begin{proof}
Because of the definition of the map $\sigma_{u,v}^{u',v'}$ it is enough to verify that the maps which are listed in the Table \ref{table 3} are Poisson maps. To this we need to verify fomulae \eqref{eq: rat 1}--\eqref{eq: rat 4} with variables $c'_j,d'_j$ and the $n$-tuple $\varepsilon'$. We will only provide justification for the first case of Table \ref{table 3}, i.e.  
$\varepsilon_i=2,\varepsilon_{i+1}=0,\,\varepsilon'_i=1,\varepsilon'_{i+1}=1$. Similar computations can be done for the remaining cases of Table \ref{table 3}. We have to consider the following four brackets.

{\bf Case 1:} $\lbrace d'_j,d'_{j+k}\rbrace_{nl}$. We have two cases to consider.
\begin{enumerate}
\item[(1)] Case: $k=1$. We have the following cases depending on the value of $j$.
\begin{enumerate}
\item[(a)] If $j=i+1$ then 
\begin{align*}
\lbrace d'_j,d'_{j+1}\rbrace_{nl}&=\left\{\frac{d_{i+1}}{1+c_i},d_{i+2}\right\}_{nl}=\frac{\lbrace d_{i+1},d_{i+2}\rbrace_{nl}}{1+c_i}-\frac{d_{i+1}}{(1+c_i)^2}\lbrace c_i,d_{i+2}\rbrace_{nl}
\stackrel{\varepsilon_{i+1}=0}{=}c_{i+1}d_{i+1}=c'_jd'_j.
\end{align*}
\item[(b)] If $j=i$ then 
\begin{align*}
\lbrace d'_j,d'_{j+1}\rbrace_{nl}&=\left\{d_i(1+c_i),\frac{d_{i+1}}{1+c_i}\right\}_{nl}=\frac{(1+c_i)\lbrace d_i,d_{i+1}\rbrace_{nl}+d_i\lbrace c_i,d_{i+1}\rbrace_{nl}}{1+c_i}-\frac{d_{i+1}}{1+c_i}\lbrace d_i,c_i\rbrace_{nl}\\
&=\frac{c_id_{i+1}}{1+c_i}=c'_jd'_j.
\end{align*}
\item[(c)] If $j=i-1$ then 
\begin{align*}
\lbrace d'_j,d'_{j+1}\rbrace_{nl}&=\lbrace d_{i-1},d_i(1+c_i)\rbrace_{nl}=d_i\lbrace d_{i-1},c_i\rbrace_{nl}+(1+c_i)\lbrace d_{i-1},d_i\rbrace_{nl}\stackrel{\varepsilon_i=2}{=}(1+c_i)c_{i-1}d_{i-1}=c'_jd'_j.
\end{align*}

\item[(d)] If $j<i-1$ or $j>i+1$ then we can see that the bracket $\lbrace d'_j,d'_{j+1}\rbrace_{nl}=\lbrace d_j,d_{j+1}\rbrace_{nl}=c_jd_j=c'_jd'_j$. 
\end{enumerate}
Therefore for any $j$ we have
\[
\lbrace d'_j,d'_{j+1}\rbrace_{nl}=c'_jd'_j.
\]

\item[(2)] Case: $k>1$. We have the following cases depending on the value of $j$.
\begin{enumerate}
\item[(a)] If $j>i+1$ then the formula \eqref{eq: rat 1} for the bracket $\lbrace d'_j,d'_{j+k}\rbrace_{nl}$ is verified straightforward.
\item[(b)] If $j=i+1$ then
\begin{align*}
\lbrace d'_j,d'_{j+k}\rbrace_{nl}&=\left\{\frac{d_{i+1}}{1+c_i},d_{i+1+k}\right\}_{nl}=\frac{\lbrace d_{i+1},d_{i+1+k}\rbrace_{nl}}{1+c_i}-\frac{d_{i+1}}{(1+c_i)^2}\lbrace c_i,d_{i+1+k}\rbrace_{nl}\\
&\stackrel{\varepsilon_{i+1}=0}{=}
\begin{cases}
d_{i+1}\prod_{l=i+1}^{i+k}c_l  &\text{ if } \varepsilon_{i+2}=\ldots=\varepsilon_{i+k}=0\\
0 & \text{ otherwise }
\end{cases}\\
&=
\begin{cases}
d'_j\prod_{l=j}^{j+k-1}c'_j  &\text{ if } \varepsilon'_{j+1}=\ldots=\varepsilon'_{j+k-1}=0\\
0 & \text{ otherwise }
\end{cases}
\end{align*}
\item[(c)] If $j=i$ then
\begin{align*}
\lbrace d'_j,d'_{j+k}\rbrace_{nl}&=\lbrace d_i(1+c_i),d_{i+k}\rbrace_{nl}=(1+c_i)\lbrace d_i,d_{i+k}\rbrace_{nl}+d_i\lbrace c_i,d_{i+k}\rbrace_{nl}\stackrel{\varepsilon_{i+1}=0}{=}0
\end{align*}
which is consistent with \eqref{eq: rat 1} because $\varepsilon'_{i+1}=1$. 
\item[(d)] If $j<i$ then we have four cases to consider.
\begin{enumerate}
\item[(i)] If $j+k>i+1$ then $\lbrace d'_j,d'_{j+k}\rbrace_{nl}=\lbrace d_j,d_{j+k}\rbrace_{nl}\stackrel{\varepsilon_i=2}{=}0$ which is consistent with equation \eqref{eq: rat 1} because $\varepsilon'_i=1$.
\item[(ii)] If $j+k=i+1$ then $\lbrace d'_j,d'_{j+k}\rbrace_{nl}=\lbrace d_j,\frac{d_{i+1}}{1+c_i}\rbrace_{nl}$. Here we have two consider two cases $j=i-1$ and $j<i-1$. Since, $\varepsilon_i=2$, both of the cases are equal to zero and the result is consistent with \eqref{eq: rat 1} because $\varepsilon'_i=1$.
\item[(iii)] If $j+k=i$ then 
\begin{align*}
\lbrace d'_j,d'_i\rbrace_{nl}
&=\lbrace d_j,d_i(1+c_i)\rbrace_{nl}\stackrel{\varepsilon_i=2}{=}(1+c_i)
\begin{cases}
d_j\prod_{l=j}^{i-1}c_l & \text{ if } \varepsilon_{j+1}=\ldots=\varepsilon_{i-1}=0\\
0 & \text{ otherwise}
\end{cases}\\
&=\begin{cases}
d'_j\prod_{l=j}^{i-1}c'_l & \text{ if } \varepsilon'_{j+1}=\ldots=\varepsilon'_{i-1}=0\\
0 & \text{ otherwise}
\end{cases}
\end{align*} 
\item[(iv)] If $j+k<i-1$ then $\lbrace d'_j,d'_{j+k}\rbrace_{nl}$ satisfies the second expression of \eqref{eq: rat 1}.
\end{enumerate}
\end{enumerate}
\end{enumerate}
Therefore the bracket $\lbrace d'_j,d'_{j+k}\rbrace_{nl}$ verifies equation \eqref{eq: rat 1} with respect to variables $c'_j,d'_j$ and the $n$-tuple $\varepsilon'$.

{\bf Case 2:} $\lbrace c'_j,d'_{j+k}\rbrace_{nl}$. We have three cases to consider.
\begin{enumerate}
\item[(1)] Case $k=0$. We have the following cases depending on the value of $j$.
\begin{enumerate}
\item[(a)] If $j=i-1$ then 
\[
\lbrace c'_j,d'_j\rbrace_{nl}=\lbrace c_{i-1}(1+c_i),d_{i-1}\rbrace_{nl}
\stackrel{\varepsilon_i=2}{=}(1+c_i)\lbrace c_{i-1},d_{i-1}\rbrace_{nl}=c_{i-1}(1+c_i)=c'_j.
\]
\item[(b)] If $j=i$ then
\begin{align*}
\lbrace c'_j,d'_j\rbrace_{nl}&=\left\{\frac{c_id_{i+1}}{d_i(1+c_i)^2},d_i(1+c_i)\right\}_{nl}=\frac{c_i(1+c_i)\lbrace d_{i+1},d_i\rbrace_{nl}+c_id_i\lbrace d_{i+1},c_i\rbrace_{nl}+d_{i+1}\lbrace c_i,d_i\rbrace_{nl}}{d_i(1+c_i)^2}\\
&=\frac{c_id_i}{d_i(1+c_i)^2}=c'_j.
\end{align*}
\item[(c)] If $j=i+1$ then 
\begin{align*}
\lbrace c'_j,d'_j\rbrace_{nl}&=\left\{c_{i+1}(1+c_i),\frac{d_{i+1}}{1+c_i}\right\}_{nl}=\lbrace c_{i+1},d_{i+1}\rbrace_{nl}-\frac{d_{i+1}}{1+c_i}\lbrace c_{i+1},c_i\rbrace_{nl}+\frac{c_{i+1}}{1+c_i}\lbrace c_i,d_{i+1}\rbrace_{nl}\\
&\stackrel{\varepsilon_{i+1}=0}{=}c_{i+1}(1+c_i)=c'_j.
\end{align*}
\item[(d)] If $j<i-1$ or $j>i+1$ then $\lbrace c'_j,d'_j\rbrace_{nl}=\lbrace c_j,d_j\rbrace_{nl}=c_j=c'_j$.
\end{enumerate}
Therefore for any $j$ we have
\[
\lbrace c'_j,d'_j\rbrace=c'_j
\]
\item[(2)] Case $k=1$. We have the following cases depending on value of $j$.
\begin{enumerate}
\item[(a)] If $j=i-1$ then 
\begin{align*}
\lbrace c'_j,d'_{j+1}\rbrace_{nl}&=\lbrace c_{i-1}(1+c_i),d_i(1+c_i)\rbrace_{nl}\stackrel{\varepsilon_i=2}{=}(1+c_i)^2\lbrace c_{i-1},d_i\rbrace_{nl}+c_{i-1}(1+c_i)\lbrace c_i,d_i\rbrace_{nl}\\
&=-c_{i-1}(1+c_i)(1+c_{i-1}(1+c_i))=-c'_j(c'_j+1).
\end{align*}
\item[(b)] If $j=i$ then 
\begin{align*}
\lbrace c'_j,d'_{j+1}\rbrace_{nl}&=\left\{\frac{c_id_i}{d_i(1+c_i)^2},\frac{d_{i+1}}{1+c_i}\right\}_{nl}=\frac{1}{d_i(1+c_i)^2}\left[\lbrace c_i,d_{i+1}\rbrace_{nl}-\frac{c_i}{d_i(1+c_i)^2}\lbrace d_i(1+c_i),d_{i+1}\rbrace_{nl}\right.\\
&\quad+\left.\frac{c_id_{i+1}}{d_i(1+c_i)^2}\lbrace d_i,c_i\rbrace_{nl}\right]=-\frac{c_id_{i+1}}{d_i(1+c_i)^2}\left(1+\frac{c_id_{i+1}}{d_i(1+c_i)^2}\right)=-c'_j(c'_j+1).
\end{align*}
\item[(c)] If $j=i+1$ then
\begin{align*}
\lbrace c'_j,d'_{j+1}\rbrace_{nl}&=\lbrace c_{i+1}(1+c_i),d_{i+2}\rbrace_{nl}=(1+c_i)\lbrace c_{i+1},d_{i+2}\rbrace_{nl}+c_{i+1}\lbrace c_i,d_{i+2}\rbrace_{nl}\\
&\stackrel{\varepsilon_{i+1}=0}{=}-c_{i+1}(1+c_i)(1+c_{i+1}(1+c_i))=-c'_j(c'_j+1).
\end{align*}
\item[(d)] If $j<i-1$ and $j>i+1$ then we can see that $\lbrace c'_j,d'_{j+1}\rbrace_{nl}$ satisfies the second case of equation \eqref{eq: rat 2}.
\end{enumerate}
Therefore for any $j$
\[
\lbrace c'_j,d'_{j+1}\rbrace_{nl}=-c'_j(c'_j+1).
\]

\item[(3)] Case $k>1$. We have the following cases depending on the value of  $j$.
\begin{enumerate}
\item[(a)] If $j>i+1$ then the last three expressions of \eqref{eq: rat 2} for the bracket $\lbrace c'_j,d'_{j+k}\rbrace_{nl}$ are verified straightforward.
\item[(b)] If $j=i+1$ then
\begin{align*}
\lbrace c'_j,d'_{j+k}\rbrace_{nl}&=\lbrace c_{i+1}(1+c_i),d_{i+k+1}\rbrace_{nl}\\
&\stackrel{\varepsilon_{i+1}=0}{=}
\begin{cases}
-(1+c_i)(c_{i+1}(c_i+1)+1)\prod_{l=i+1}^{i+k}c_l & \text{ if } \varepsilon_{i+2}=\ldots=\varepsilon_{i+k}=0,\\
(1+c_i)\prod_{l=i+1}^{i+k}c_l & \text{ if } \varepsilon_{i+2}=2,\varepsilon_{i+3}=\ldots=\varepsilon_{i+k}=0,\\
0 & \text{ otherwise},
\end{cases}\\
&=
\begin{cases}
-(c'_j+1)\prod_{l=j}^{j+k-1}c'_l & \text{ if } \varepsilon'_{j+1}=\ldots=\varepsilon'_{j+k-1}=0,\\
\prod_{l=j}^{j+k-1}c'_l & \text{ if } \varepsilon'_{j+1}=2,\varepsilon'_{j+2}=\ldots=\varepsilon'_{j+k-1}=0,\\
0 & \text{ otherwise}.
\end{cases}
\end{align*}
\item[(c)] If $j=i$ then
\begin{align*}
\lbrace c'_j,d'_{j+k}\rbrace_{nl}&=\left\{\frac{c_id_{i+1}}{d_i(1+c_i)^2},d_{i+k}\right\}_{nl}=\frac{\frac{d_{i+1}(1-c_i)}{1+c_i}\lbrace c_i,d_{i+k}\rbrace_{nl}+c_i\lbrace d_{i+1},d_{i+k}\rbrace_{nl}-\frac{c_id_{i+1}}{d_i}\lbrace d_i,d_{i+k}\rbrace_{nl}
}{d_i(1+c_i)^2}\\
&\stackrel{\varepsilon_{i+1}=0}{=}-\frac{d_{i+1}}{d_i(1+c_i)^2}(1-c_i-1+c_i)
\begin{cases}
\prod_{l=i}^{i+k-1}c_l &\text{ if } \varepsilon_{i+2}=\ldots=\varepsilon_{i+k-1}=0,\\
0 & \text{ otherwise },
\end{cases}\\
&=0.
\end{align*}
which is consistent with \eqref{eq: rat 2} because $\varepsilon'_{i+1}=1$.\\
\item[(d)] If $j=i-1$, we have that $\lbrace c'_j,d'_{j+k}\rbrace_{nl}=0$ which is compatible with equation \eqref{eq: rat 2} since $\varepsilon'_i=1$. In fact,
\begin{enumerate}
\item[(i)] If $k=2$ then
\[
\lbrace c'_j,d'_{j+k}\rbrace_{nl}=\left\{ c_{i-1}(1+c_i),\frac{d_{i+1}}{1+c_i}\right\}_{nl}\stackrel{\varepsilon_i=2}{=}0.
\]
\item[(ii)] If $k>2$ then
\begin{align*}
\lbrace c'_j,d'_{j+k}\rbrace_{nl}&=\lbrace c_{i-1}(1+c_i),d_{i-1+k}\rbrace_{nl}=c_{i-1}\lbrace c_i,d_{i-1+k}\rbrace_{nl}+(1+c_i)\lbrace c_{i-1},d_{i-1+k}\rbrace_{nl}\equals_{\varepsilon_{i+1}=0}^{\varepsilon_i=2}0.
\end{align*}
\end{enumerate}
\item[(e)] If $j<i-1$ then we have three cases to consider.
\begin{enumerate}
\item[(i)] If $j+k=i$ then 
\begin{align*}
\lbrace c'_j,d'_{j+k}\rbrace_{nl}&=\lbrace c_j,d_i(1+c_i)\rbrace_{nl}\stackrel{\varepsilon_i=2}{=}(1+c_i)
\begin{cases}
-(c_j+1)\prod_{l=j}^{i-1}c_l &\text{ if } \varepsilon_{j+1}=\ldots=\varepsilon_{i-1}=0,\\
\prod_{l=j}^{i-1}c_l &\text{ if } \varepsilon_{j+1}=2,\varepsilon_{j+2}\ldots=\varepsilon_{i-1}=0,\\
0 &\text{ otherwise},
\end{cases}\\
&=
\begin{cases}
-(c'_j+1)\prod_{l=j}^{j+k-1}c'_l &\text{ if } \varepsilon'_{j+1}=\ldots=\varepsilon'_{j+k-1}=0,\\
\prod_{l=j}^{j+k-1}c'_l &\text{ if } \varepsilon'_{j+1}=2,\varepsilon'_{j+2}=\ldots=\varepsilon'_{j+k-1}=0,\\
0 &\text{ otherwise}.
\end{cases}
\end{align*}
\item[(ii)] If $j+k=i+1$ then
\[
\lbrace c'_j,d'_{j+k}\rbrace_{nl}=\left\{c_j,\frac{d_{i+1}}{1+c_i}\right\}_{nl}=\frac{\lbrace c_j,d_{i+1}\rbrace_{nl}}{1+c_i}-\frac{d_{i+1}}{(1+c_i)^2}\lbrace c_j,c_i\rbrace_{nl}\stackrel{\varepsilon_i=2}{=}0
\]
which is consistent with equation \eqref{eq: rat 2} since $\varepsilon'_i=1$.
\item[(iii)] If $j+k<i-1$ or $j+k>i+1$  then we can see immediately that $\lbrace c'_j,'d'_{j+k}\rbrace_{nl}$ satisfies one of the last three cases of \eqref{eq: rat 2} depending on the $n$-tuple $\varepsilon'$.
\end{enumerate}
\end{enumerate}
\end{enumerate}
Summarizing, the bracket $\lbrace c'_j,d'_{j+k}\rbrace_{nl}$ satisfies the equation \eqref{eq: rat 2} with respect to variables $c'_j,d'_j$ and the $n$-tuple $\varepsilon'$.

{\bf Case 3:} $\lbrace c'_{j+k},d'_j\rbrace_{nl}$. We have two cases to consider.
\begin{enumerate}
\item[(1)] Case $k=1$. We have the following cases depending on the value of $j$.
\begin{enumerate}
\item[(a)] If $j=i-2$ then
\begin{align*}
\lbrace c'_{j+1},d'_j\rbrace_{nl}&=\lbrace c_{i-1}(1+c_i),d_{i-2}\rbrace_{nl}\stackrel{\varepsilon_i=2}{=}(1+c_i)\lbrace c_{i-1},d_{i-2}\rbrace_{nl}=(1+c_i)(2-\varepsilon_{i-1})c_{i-2}c_{i-1}\frac{d_{i-2}}{d_{i-1}}\\
&=(2-\varepsilon'_{j+1})c'_jc'_{j+1}\frac{d'_j}{d'_{j+1}}.
\end{align*} 
\item[(b)] If $j=i-1$ then
\begin{align*}
\lbrace c'_{j+1},d'_j\rbrace_{nl}&=\left\{\frac{c_id_{i+1}}{d_i(1+c_i)^2},d_{i-1}\right\}_{nl}\stackrel{\varepsilon_i=2}{=}-\frac{c_id_{i+1}}{d_i^2(1+c_i)^2}\lbrace d_i,d_{i-1}\rbrace_{nl}=\frac{c_id_{i+1}c_{i-1}d_{i-1}}{d_i^2(1+c_i)^2}=c'_jc'_{j+1}\frac{d'_j}{d'_{j+1}}.
\end{align*} 
which is consistent with the first case of equation \eqref{eq: rat 3} since $\varepsilon'_i=1$.
\item[(c)] If $j=i$ then
\begin{align*}
\lbrace c'_{j+1},d'_j\rbrace_{nl}&=\lbrace c_{i+1}(1+c_i),d_i(1+c_i)\rbrace_{nl}=(1+c_i)\left[(1+c_i)\lbrace c_{i+1},d_i\rbrace_{nl}+d_i\lbrace c_{i+1},c_i\rbrace_{nl}+c_{i+1}\lbrace c_i,d_i\rbrace_{nl}\right]\\
&\stackrel{\varepsilon_{i+1}=0}{=}(1+c_i)c_ic_{i+1}=c'_jc'_{j+1}\frac{d'_j}{d'_{j+1}}.
\end{align*}
which is consistent with the first case of equation \eqref{eq: rat 3} since $\varepsilon'_{i+1}=1$.
\item[(d)] If $j=i+1$ then
\begin{align*}
\lbrace c'_{j+1},d'_j\rbrace_{nl}&=\left\{c_{i+2},\frac{d_{i+1}}{1+c_i}\right\}_{nl}=\frac{\lbrace c_{i+2},d_{i+1}\rbrace_{nl}}{1+c_i}-\frac{d_{i+1}}{(1+c_i)^2}\lbrace c_{i+2},c_i\rbrace_{nl}\\
&\stackrel{\varepsilon_{i+1}=0}{=}(2-\varepsilon_{i+2})c_{i+1}c_{i+2}\frac{d_{i+1}}{d_{i+2}}=(2-\varepsilon'_{j+1})c'_jc'_{j+1}\frac{d'_j}{d'_{j+1}}.
\end{align*}
\item[(e)] If $j>i+1$ or $j<i-2$ then we can see that the bracket $\lbrace c'_{j+1},d'_j\rbrace_{nl}$ satisfies the first case of equation \eqref{eq: rat 3}. 
\end{enumerate}
Therefore for any $j$ we have
\[
\lbrace c'_{j+1},d'_j\rbrace_{nl}=(2-\varepsilon'_{j+1})c'_jc'_{j+1}\frac{d'_j}{d'_{j+1}}.
\]
\item[(2)] Let $k>1$. We have the following cases depending in the values of $j$.
\begin{enumerate}
\item[(a)] If $j>i+1$ then the last two expressions of \eqref{eq: rat 3} for the bracket $\lbrace c'_{j+k},d'_j\rbrace_{nl}$ are verified straightforward.
\item[(b)] If $j=i+1$ then
\begin{align*}
\lbrace c'_{j+k},d'_j\rbrace_{nl}&=\left\{ c_{i+1+k},\frac{d_{i+1}}{1+c_i}\right\}_{nl}=\frac{\lbrace c_{i+1+k},d_{i+1}\rbrace_{nl}}{1+c_i}-\frac{d_{i+1}}{(1+c_i)^2}\lbrace c_{i+1+k},c_i\rbrace_{nl}\\
&\stackrel{\varepsilon_{i+1}=0}{=}
\begin{cases}
(2-\varepsilon_{i+1+k})\frac{d_{i+1}}{d_{i+1+k}}\prod_{l=i+1}^{i+1+k}c_l &\text{ if } \varepsilon_{i+2}=\ldots=\varepsilon_{i+k}=0,\\
0 &\text{ otherwise},
\end{cases}\\
&=
\begin{cases}
(2-\varepsilon'_{j+k})\frac{d'_j}{d'_{j+k}}\prod_{l=j}^{j+k}c'_l &\text{ if } \varepsilon'_{j+1}=\ldots=\varepsilon'_{j+k-1}=0,\\
0 &\text{ otherwise}.
\end{cases}
\end{align*}
\item[(c)] If $j=i$ then
\[
\lbrace c'_{j+k},d'_j\rbrace_{nl}=\lbrace c_{i+k},d_i(1+c_i)\rbrace_{nl}=(1+c_i)\lbrace c_{i+k},d_i\rbrace_{nl}+d_i\lbrace c_{i+k},c_i\rbrace_{nl}\stackrel{\varepsilon_{i+1}=0}{=}0,
\]
which is consistent with \eqref{eq: rat 3} because $\varepsilon'_{i+1}=1$.
\item[(d)] If $j<i$ we have the following subcases
\begin{enumerate}
\item[(i)] If $j+k<i-1$ or $j+k>i+1$ then the equation \eqref{eq: rat 3} is verified easily.
\item[(ii)] If $j+k=i-1$ then
\begin{align*}
\lbrace c'_{j+k},d'_j\rbrace_{nl}&=\lbrace c_{i-1}(1+c_i),d_j\rbrace_{nl}\stackrel{\varepsilon_i=2}{=}(1+c_i)
\begin{cases}
(2-\varepsilon_{i-1})\frac{d_j}{d_{i-1}}\prod_{l=j}^{i-1}c_l & \text{ if } \varepsilon_{j+1}=\ldots=\varepsilon_{i-2}=0,\\
0 & \text{ otherwise},
\end{cases}\\
&=
\begin{cases}
(2-\varepsilon'_{j+k})\frac{d'_j}{d'_{j+k}}\prod_{l=j}^{j+k}c'_l & \text{ if } \varepsilon'_{j+1}=\ldots=\varepsilon'_{j+k-1}=0,\\
0 & \text{ otherwise}.
\end{cases}
\end{align*}
\item[(iii)] If $j+k=i$ then
\begin{align*}
\lbrace c'_{j+k},d'_j\rbrace_{nl} &=\left\{\frac{c_id_{i+1}}{d_i(1+c_i)^2},d_j\right\}_{nl}\stackrel{\varepsilon_i=2}{=}-\frac{c_id_{i+1}}{d_i^2(1+c_i)^2}\lbrace d_i,d_j\rbrace_{nl}\\
&=\frac{c_id_{i+1}}{d_i^2(1+c_i)^2}
\begin{cases}
-d_j\prod_{l=j}^{i-1}c_l &\text{ if } \varepsilon_{j+1}=\ldots=\varepsilon_{i-1}=0,\\
0 & \text{ otherwise},
\end{cases}\\
&=
\begin{cases}
\frac{d'_j}{d'_{j+k}}\prod_{l=j}^{j+k}c'_l & \text{ if } \varepsilon'_{j+1}=\ldots=\varepsilon'_{j+k-1}=0,\\
0 & \text{ otherwise}.
\end{cases}
\end{align*}
\item[(iv)] If $j+k=i+1$ then $\lbrace c'_{j+k},d'_j\rbrace_{nl}=\lbrace c_{i+1}(c_i+1),d_j\rbrace_{nl}\stackrel{\varepsilon_i=2}{=}0$ which is consistent with the last two cases of \eqref{eq: rat 3} because $\varepsilon'_i=1$.
\end{enumerate}
\end{enumerate}
\end{enumerate}
Summarizing, the bracket $\lbrace c'_{j+k},d'_j\rbrace_{nl}$ satisfies the equation \eqref{eq: rat 3} with respect to the variables $c'_j,d'_j$ and the $n$-tuple $\varepsilon'$.

{\bf Case 4:} $\lbrace c'_j,c'_{j+k}\rbrace_{nl}$. We have two cases to consider.
\begin{enumerate}
\item[(1)] Case $k=1$. We have the following cases depending on the values of $j$. 
\begin{enumerate}
\item[(a)] If $j=i-2$ then
\begin{align*}
\lbrace c'_j,c'_{j+1}\rbrace_{nl}&=\lbrace c_{i-2},c_{i-1}(1+c_i)\rbrace_{nl}\stackrel{\varepsilon_i=2}{=}(1+c_i)(2-\varepsilon_{i-1})\frac{c_{i-2}+1}{d_{i-1}}c_{i-2}c_{i-1}=(2-\varepsilon'_{j+1})\frac{c'_j+1}{d'_{j+1}}c'_jc'_{j+1}.
\end{align*}
\item[(b)] If $j=i-1$ then
\begin{align*}
\lbrace c'_j,c'_{j+1}\rbrace_{nl}&=\left\{c_{i-1}(1+c_i),\frac{c_id_{i+1}}{d_i(1+c_i)^2}\right\}_{nl}\stackrel{\varepsilon_i=2}{=}-\frac{c_id_{i+1}}{d_i^2(1+c_i)^2}\left[(1+c_i)\lbrace c_{i-1},d_i\rbrace_{nl}+c_{i-1}\lbrace c_i,d_i\rbrace_{nl}\right]\\
&=\frac{c_id_{i+1}}{d_i^2(1+c_i)^2}c_{i-1}(1+c_{i-1}(1+c_i))=\frac{c'_j+1}{d'_{j+1}}c'_jc'_{j+1},
\end{align*}
which is consistent with the first case of equation \eqref{eq: rat 4} because $\varepsilon'_i=1$.
\item[(c)] If $j=i$ then
\begin{align*}
\lbrace c'_j,c'_{j+1}\rbrace_{nl}&=\left\{\frac{c_id_{i+1}}{d_i(1+c_i)^2},c_{i+1}(1+c_i),\right\}_{nl}\stackrel{\varepsilon_{i+1}=0}{=}\frac{c_ic_{i+1}}{d_i}\left(1+\frac{c_id_{i+1}}{d_i(1+c_i)^2}\right)=\frac{c'_j+1}{d'_{j+1}}c'_jc'_{j+1},
\end{align*}
which is consistent with the first case of equation \eqref{eq: rat 4} since $\varepsilon'_{i+1}=1$.
\item[(d)] If $j=i+1$ then
\begin{align*}
\lbrace c'_j,c'_{j+1}\rbrace_{nl}&=\lbrace c_{i+1}(1+c_i),c_{i+2}\rbrace_{nl}\stackrel{\varepsilon_{i+1}=0}{=}(2-\varepsilon_{i+2})\frac{(1+c_{i+1}(1+c_i))}{d_{i+2}}c_{i+1}(c_i+1)c_{i+2}\\
&=(2-\varepsilon'_{j+1})\frac{c'_j+1}{d'_{j+1}}c'_jc'_{j+1}.
\end{align*}
\item[(e)] If $j>i+1$ or $j<i-2$ then it is straightforward that the bracket $\lbrace c'_j,c'_{j+1}\rbrace_{nl}$ satisfies the first case of \eqref{eq: rat 3}.
\end{enumerate}
Therefore for any $j$ we have
\[
\lbrace c'_j,c'_{j+1}\rbrace_{nl}=(2-\varepsilon'_{j+1})\frac{c'_j+1}{d'_{j+1}}c'_jc'_{j+1}.
\]
\item[(2)] Case $k>1$. We have the following cases depending on the values of $j$.
\begin{enumerate}
\item[(a)] If $j>i+1$ then the last three expressions of \eqref{eq: rat 4} for the bracket $\lbrace c'_j,c'_{j+k}\rbrace_{nl}$ are verified straightforward. 
\item[(b)] If $j=i+1$ then
\begin{align*}
\lbrace c'_j,c'_{j+k}\rbrace_{nl}&=\lbrace c_{i+1}(1+c_i),c_{i+1+k}\rbrace_{nl}=(1+c_i)\lbrace c_{i+1},c_{i+1+k}\rbrace_{nl}+c_{i+1}\lbrace c_i,c_{i+1+k}\rbrace_{nl}\\
&\stackrel{\varepsilon_{i+1}=0}{=}
\begin{cases}
(2-\varepsilon_{i+1+k})\frac{c_{i+1}(1+c_i)+1}{d_{i+1+k}}(1+c_i)\prod_{l=i+1}^{i+k+1}c_l & \text{ if } \varepsilon_{i+2}=\ldots=\varepsilon_{i+k}=0,\\
(\varepsilon_{i+1+k}-2)\frac{1+c_i}{d_{i+1+k}}\prod_{l=i+1}^{i+k+1}c_l & \text{ if }\varepsilon_{i+2}=2,\varepsilon_{i+3}=\ldots=\varepsilon_{i+k}=0,\\
0 &\text{ otherwise},
\end{cases}\\
&=
\begin{cases}
(2-\varepsilon_{j+k})\frac{c'_j+1}{d'_{j+k}}\prod_{l=j}^{j+k}c'_l & \text{ if } \varepsilon'_{j+1}=\ldots=\varepsilon'_{j+k-1}=0,\\
(\varepsilon_{j+k}-2)\frac{1}{d'_{j+k}}\prod_{l=j}^{j+k}c'_l & \text{ if }\varepsilon'_{j+1}=2,\varepsilon'_{j+2}=\ldots=\varepsilon'_{j+k-1}=0,\\
0 &\text{ otherwise}.
\end{cases}
\end{align*}
\item[(c)] If $j=i$ then
\begin{align*}
\lbrace c'_j,c'_{j+k}\rbrace_{nl}&=\left\{\frac{c_id_{i+1}}{d_i(1+c_i)^2},c_{i+k}\right\}_{nl}\\
&\stackrel{\varepsilon_{i+1}=0}{=}\frac{1}{d_i(1+c_i)^2}
\begin{cases}
(2-\varepsilon_{i+k})\frac{c_id_{i+1}}{d_{i+k}}\prod_{l=i}^{i+k}c_l & \text{ if } \varepsilon_{i+2}=\ldots=\varepsilon_{i+k-1}
=0,\\
0 & \text{ otherwise},
\end{cases}\\
&\mathrel{\phantom{=}}-\frac{c_id_{i+1}}{d_i^2(1+c_i)^2}
\begin{cases}
(2-\varepsilon_{i+k})\frac{d_i(c_i+1)}{d_{i+k}}\prod_{l=i}^{i+k}c_l & \text{ if } \varepsilon_{i+2}=\ldots=\varepsilon_{i+k-1}
=0,\\
0 & \text{ otherwise}.
\end{cases}\\
&=0.
\end{align*}
which is consistent with \eqref{eq: rat 4} beacuse $\varepsilon'_{i+1}=0$.
\item[(d)] If $j=i-1$ then $\lbrace c'_j,c'_{j+k}\rbrace_{nl}=0$ which is consistent with \eqref{eq: rat 4} because 
$\varepsilon'_i=0$. In fact,
\begin{enumerate}
\item[(i)]  If $k=2$ then
\begin{align*}
\lbrace c'_j,c'_{j+k}\rbrace_{nl}&=\lbrace c_{i-1}(1+c_i),c_{i+1}(1+c_i)\rbrace_{nl}\stackrel{}{=}(1+c_i)^2\lbrace c_{i-1},c_{i+1}\rbrace_{nl}+c_{i-1}(1+c_i)\lbrace c_i,c_{i+1}\rbrace_{nl}\\
&+c_{i+1}(1+c_i)\lbrace c_{i-1},c_i\rbrace_{nl}\equals^{\varepsilon_i=2}_{\varepsilon_{i+1}=0}0.
\end{align*}
\item[(ii)] If $k>2$ then
\begin{align*}
\lbrace c'_j,c'_{j+k}\rbrace_{nl}&=\lbrace c_{i-1}(1+c_i),c_{i-1+k}\rbrace_{nl}\\
&\equals^{\varepsilon_i=2}_{\varepsilon_{i+1}=0}c_{i-1} 
\begin{cases}
(2-\varepsilon_{i-1+k})\frac{c_i+1}{d_{i-1+k}}\prod_{l=i}^{i-1+k}c_l &\text{ if }\varepsilon_{i+2}=\ldots=\varepsilon_{i-2+k}
=0,\\
0 & \text{ otherwise},
\end{cases}\\
&\mathrel{\phantom{=}}+(1+c_i)
\begin{cases}
(\varepsilon_{i-1+k}-2)\frac{1}{d_{i-1+k}}\prod_{l=i-1}^{i-1+k}c_l &\text{ if }\varepsilon_{i+2}=\ldots=\varepsilon_{i-2+k}
=0,\\
0 & \text{ otherwise}.
\end{cases}\\
&=0.
\end{align*}
\end{enumerate}
\item[(e)] If $j<i-1$ then we have the following subcases
\begin{enumerate}
\item[(i)] If $j+k<i-1$ or $j+k>i+1$ the bracket $\lbrace c'_j,c'_{j+k}\rbrace_{nl}$ satisfies equation \eqref{eq: rat 4}.
\item[(ii)] If $j+k=i-1$ then
\begin{align*}
\lbrace c'_j,c'_{j+k}\rbrace_{nl}&=\lbrace c_j,c_{i-1}(1+c_i)\rbrace_{nl}\stackrel{\varepsilon_i=2}{=}(1+c_i)\lbrace c_j,c_{i-1}
\rbrace_{nl}\\
&=(1+c_i)
\begin{cases}
(2-\varepsilon_{i-1})\frac{c_j+1}{d_{i-1}}\prod_{l=j}^{i-1}c_l & \text{ if } \varepsilon_{j+1}=\ldots=\varepsilon_{i-2}=0,\\
(\varepsilon_{i-1}-2)\frac{1}{d_{i-1}}\prod_{l=j}^{i-1}c_l & \text{ if } \varepsilon_{j+1}=2,\varepsilon_{j+2}=\ldots=\varepsilon_{i-2}=0,\\
0 & \text{ otherwise}.
\end{cases}\\
&=
\begin{cases}
(2-\varepsilon'_{j+k})\frac{c'_j+1}{d'_{j+k}}\prod_{l=j}^{j+k}c'_l & \text{ if } \varepsilon'_{j+1}=\ldots=\varepsilon'_{j+k-1}=0,\\
(\varepsilon'_{j+k}-2)\frac{1}{d'_{j+k}}\prod_{l=j}^{j+k}c'_l & \text{ if } \varepsilon'_{j+1}=2,\varepsilon'_{j+2}=\ldots=\varepsilon'_{j+k-1}=0,\\
0 & \text{ otherwise}.
\end{cases}
\end{align*}
\item[(iii)] If $j+k=i$ then
\begin{align*}
\lbrace c'_j,c'_{j+k}\rbrace_{nl}&=\left\{c_j,\frac{c_id_{i+1}}{d_i(1+c_i)^2}\right\}_{nl}\stackrel{\varepsilon_i=2}{=}-\frac{c_id_{i+1}}{d_i^2(1+c_i)^2}\lbrace c_j,d_i\rbrace_{nl}\\
&=\frac{c_id_{i+1}}{d^2_i(1+c_i)^2}
\begin{cases}
-(c_j+1)\prod_{l=j}^{i-1}c_l & \text{ if } \varepsilon_{j+1}=\ldots=\varepsilon_{i-1}=0,\\
\prod_{l=j}^{i-1}c_l & \text{ if } \varepsilon_{j+1}=2,\varepsilon_{j+2}=\ldots=\varepsilon_{i-1}=0,\\
0 & \text{ otherwise},
\end{cases}\\
&=
\begin{cases}
\frac{c'_j+1}{d'_{j+k}}\prod_{l=j}^{j+k}c'_l & \text{ if } \varepsilon'_{j+1}=\ldots=\varepsilon'_{j+k-1}=0,\\
-\frac{1}{d'_{j+k}}\prod_{l=j}^{j+k}c'_l & \text{ if } \varepsilon'_{j+1}=2,\varepsilon'_{j+2}=\ldots=\varepsilon'_{j+k-1}=0,\\
0 & \text{ otherwise}.
\end{cases}
\end{align*}
which is compatible with equation \eqref{eq: rat 4} since $\varepsilon'_i=1$.
\item[(iv)] If $j+k=i+1$ then $\lbrace c'_j,c'_{j+k}\rbrace_{nl}=\lbrace c_j,c_{i+1}(1+c_i)\rbrace_{nl}\stackrel{\varepsilon_i=2}{=}0$ which is consistent with \eqref{eq: rat 4} because $\varepsilon'_i=1$.
\end{enumerate}
\end{enumerate}
\end{enumerate}
Summarizing, the bracket $\lbrace c'_j,c'_{j+k}\rbrace_{nl}$ satisfies the equation \eqref{eq: rat 4} with respect to the variables $c'_j,d'_j$ and the $n$-tuple $\varepsilon'$.
\end{proof}

\subsection*{Acknowledgements}
First, I am very thankful to M. Gekhtman for his support, guidance and stimulating conversations during my visit to the University of Notre Dame. Second, I appreciate the suggestions of M. Gekhtman and I. Mencattini during the writing process of this manuscript and the anonymous referees whose suggestions helped improve this pape. Lastly, I thank to FAPESP grant 2014/08512-3 and PNPD CAPES.


\begin{thebibliography}{99}


\bibitem[AH]{A-H}
M.~F.~Atiyah, N.~Hitchin, \emph{The Geometry and Dynamics of Magnetic Monopoles}, M.B. Porter
Lectures, Princeton University Press, Princeton, NJ, 1988

\bibitem[E]{Eber}
E.~Chu\~no, \emph{On Coxeter-Toda lattices}, PhD Thesis.

\bibitem[FG1]{FG1}
L.~Faybusovich and M.~Gekhtman, \emph{Elementary Toda orbits and integrable lattices}, J. Math. Phys. {\bf 41} (2000), 2905-2921.

\bibitem[FG2]{FG2}
L.~Faybusovich and M.~Gekhtman, \emph{Poisson brackets on rational functions and multi-Hamiltonian structure for integrable lattices}, Phys. Lett. A {\bf 272} (2000), 236-244.

\bibitem[FG3]{FG3}
L.~Faybusovich, M.~Gekhtman, \emph{Inverse moment problem for elementary co-adjoint orbits}, Inverse Problems {\bf 17} (2001), 1295-1306.

\bibitem[FZ]{Fo-Ze}
S.~Fomin and A.~Zelevinsky, \emph{Double Bruhat cells and total positivity}, J. Amer. Math. Soc. {\bf 12} (1999), no. 2, 335-380.

\bibitem[GSV1]{GSV1}
M.~Gekhtman, M.~Shapiro and A.~Vainshtein, \emph{Poisson geometry of directed networks in a disk}, Sel. Math., New ser. {\bf 15} (2009) 61-103.


\bibitem[GSV2]{GSV2}
M.~Gekhtman, M.~Shapiro and A.~Vainshtein, \emph{Poisson geometry of directed networks in an annulus}, J. Eur. Math. Soc. {\bf 14} (2012) 541-570.

\bibitem[GSV3]{GSV4}
M.~Gekhtman, M.~Shapiro and A.~Vainshtein, \emph{Generalized B\"acklund-Darboux transformations for Coxeter-Toda lattices from cluster algebra perspective}, Acta Math. {\bf 206} (2011), 245-310.

\bibitem[HKKR]{HKKR}
T.~Hoffmann, J.~Kellendonk, N.~Kutz and N.~Reshetikhin, \emph{Factorization dynamics and Coxeter-Toda lattices}, Comm. Math. Phys. {\bf 212} (2000), 297-321.

\bibitem[KZ]{KoZ}
M.~Kogan and A.~Zelevinsky, \emph{On symplectic leaves and integrable systems in standard complex semisimple Poisson-Lie groups}, Internat. Math. Res. Notices. (2002), 1685-1702. 

\bibitem[M]{Mo}
J.~Moser, \emph{Finitely many mass points on the line under the influence of the exponential potential-an integrable system}. Dynamical systems, theory and applications, 467-497, Lecture Notes in Physics {\bf 38}, Springer, Berlin, 1975.

\bibitem[Re]{Resh}
N.~Reshetikhin, \emph{Integrability of characteristic Hamiltonians systems on simple Lie groups with standard Poisson-Lie structure}, Comm. Mat. Phys. {\bf 242} (2003), 1-29.

\bibitem[Ru]{Ruis}
S.~N.~M.~Ruijsenaars, \emph{Relativistic Toda systems}, Comm. Math. Phys. {\bf 133}, 217-247 (1990).

\bibitem[S]{Su}
Y.B.~Suris, \emph{On the bi-Hamiltonian structure of Toda and relativistic Toda lattice}, Phys. Letters A {\bf 180} (1993), 419-429.

\bibitem[T]{To}
M.~Toda, \emph{Vibration of a chain with a non-linear interaction}, J. Phys. Soc. Japan, {\bf 22} (1967), 431-436.

\end{thebibliography}
\end{document}